\documentclass[final]{llncs}

\newif\iflong
\newif\ifshort

\longtrue
\shortfalse

\usepackage[T1]{fontenc}
\usepackage{hyperref}
\usepackage{wrapfig}

\usepackage{epsfig}
\usepackage{gastex}
\usepackage{amsmath,amssymb}
\usepackage{xspace}
\usepackage{dsfont}
\usepackage[mathscr]{euscript}
\usepackage{cite}
\usepackage{theorem}
\usepackage{hhline}
\usepackage{enumerate}

 \makeatletter
 \renewcommand{\paragraph}{\@startsection{paragraph}{6}{\z@}{2ex}{-0.7em}{\normalsize\bf}}
 \makeatother

\begin{document}


\newcommand{\Alphabet}{\Sigma}
\newcommand{\Props}{P}
\newcommand{\Int}[2]{#1^#2}
\newcommand{\rel}{\lhd}
\newcommand{\Data}{\mathfrak{D}}

\newcommand{\RelSymb}{\sigma}
\newcommand{\RelInt}{\mathfrak{I}}

\newcommand{\graph}{G}

\newcommand{\EMSOtwo}{\extEMSO_2(\SimSign \cup \{<\})}
\newcommand{\ClassEMSOtwo}{\extClassEMSO_2(\SimSign \cup \{<\})}

\newcommand{\tsucc}{+1}
\newcommand{\tequiv}{\sim}
\newcommand{\tequivp}[1]{\mathsf{equ}^{k}}

\newcommand{\req}{\mathsf{r}}
\newcommand{\ack}{\mathsf{a}}

\newcommand{\eoex}{\hfill$\Diamond$}

\newcommand{\nrel}{\rel_\Gamma}

\newcommand{\nMSO}{\MSO^+}
\newcommand{\nClassMSO}{\ClassMSO^+}

\newcommand{\newhat}[1]{#1'}

\newcommand{\wsimulate}[4]{#3 \mathrel{\sqsubseteq^{#1}_{#2}} #4}
\newcommand{\conn}[1]{\mathrel{\leftrightarrow}^{#1}}
\newcommand{\entf}{e}
\newcommand{\range}{z}

\newcommand{\io}{\alpha}
\newcommand{\jo}{\beta}

\newcommand{\myhat}[1]{#1_\Gamma}

\newcommand{\sphere}{S}
\newcommand{\core}{\mathit{core}}
\newcommand{\esphere}{E}
\newcommand{\sdist}{\mathit{dist}}
\newcommand{\scenter}{\gamma}
\newcommand{\sactive}{\alpha}
\newcommand{\scolor}{\mathit{col}}
\newcommand{\slabel}{\mathit{label}}
\newcommand{\sguard}{\mathit{data}}

\newcommand{\Part}[1]{\mathit{Part}(#1)}

\newcommand{\MSCs}{\mathsf{MSC}}

\newcommand{\mydist}{\mathit{dist}}

\newcommand{\myautomata}{class register automata\xspace}
\newcommand{\myautomaton}{class register automaton\xspace}

\newcommand{\ngautomata}{non-guessing class register automata\xspace}
\newcommand{\ngautomaton}{non-guessing class register automaton\xspace}
\newcommand{\emphngautomata}{\emph{non-guessing class register automata}\xspace}
\newcommand{\emphngautomaton}{\emph{non-guessing class register automaton}\xspace}

\newcommand{\cmautomata}{class memory automata\xspace}
\newcommand{\cmautomaton}{class memory automaton\xspace}

\newcommand{\emphmyautomaton}{\emph{class register automaton}\xspace}
\newcommand{\emphcmautomaton}{\emph{class memory automaton}\xspace}

\newcommand{\Guards}{\mathit{Guards}}

\newcommand{\mytrue}{\mathit{true}}
\newcommand{\myfalse}{\mathit{false}}

\newcommand{\ptype}[1]{\mathit{type}^-(#1)}
\newcommand{\ntype}[1]{\mathit{type}^+(#1)}

\newcommand{\N}{\mathds{N}}

\newcommand{\Acc}{\Phi}
\newcommand{\source}{p}
\newcommand{\test}{g}
\newcommand{\update}{f}

\newcommand{\set}[1]{[#1]}

\newcommand{\lformula}[2]{\letter{#1} = #2}

\newcommand{\neglformula}[2]{\letter{#1} \neq #2}

\newcommand{\init}{\iota}

\newcommand{\dvec}{\bar{d}}

\newcommand{\param}{m}

\newcommand{\symb}{\alpha}

\newcommand{\guess}{\mathsf{guess}}

\newcommand{\letter}[1]{\ell(#1)}
\newcommand{\dvalue}[2]{d^{#2}(#1)}

\newcommand{\letterp}[2]{\mathbb{l}^{#1}(#2)}
\newcommand{\dvaluep}[3]{\mathbb{d}^{#1}(#2,#3)}

\newcommand{\Sign}{\mathscr{S}}
\newcommand{\sign}{\sigma}

\newcommand{\Aut}{\mathcal{A}}
\newcommand{\Baut}{\mathcal{B}}

\newcommand{\labelone}{\lambda}
\newcommand{\labeltwo}{\nu}
\newcommand{\guard}{g}

\newcommand{\Partition}[1]{\mathit{Equ}(#1)}

\newcommand{\Local}{F}
\newcommand{\Global}{G}

\newcommand{\Trans}{\Delta}

\newcommand{\egal}{\ast}

\newcommand{\nd}{\ell}

\newcommand{\data}{d}

\newcommand{\radius}{B}

\newcommand{\reg}{\rho}
\newcommand{\Reg}{R}

\newcommand{\noentry}{\bot}

\newcommand{\edge}[1]{\prec_{#1}}
\newcommand{\edgep}[2]{\prec_{#2}^{#1}}
\newcommand{\fedge}[1]{\prec_{#1}}
\newcommand{\fedgep}[2]{\sqsubset_{#2}^{#1}}

\newcommand{\equ}[1]{\sim_{#1}}

\newcommand{\edgesim}{\prec_\sim}
\newcommand{\edgesucc}{\prec_{+1}}

\newcommand{\before}[1]{<_{#1}}

\newcommand{\pos}[2]{\#1^{#1}}

\newcommand{\Spheres}[1]{#1\textup{-}\mathit{Spheres}_\Sign}
\newcommand{\eSpheres}[1]{#1\textup{-}\mathit{eSpheres}_\Sign}
\newcommand{\mineSpheres}[1]{\mathit{eSpheres}_{#1}^\textup{min}}

\newcommand{\Sph}[3]{#1\textup{-}\mathit{Sph}^{#2}(#3)}
\newcommand{\Sphp}[4]{#1\textup{-}\mathit{Sph}^{#2}(#3)}

\newcommand{\trans}[3]{#1 \stackrel{#2}{\longrightarrow} #3}

\newcommand{\Graph}[1]{\graph(#1)}

\newcommand{\loc}{{\textup{r}}}

\newcommand{\MSO}{\loc\textup{MSO}}
\newcommand{\extMSO}{\textup{MSO}}
\newcommand{\EMSO}{\loc\textup{EMSO}}
\newcommand{\extEMSO}{\textup{EMSO}}
\newcommand{\FO}{\loc\textup{FO}}
\newcommand{\extFO}{\textup{FO}}

\newcommand{\ClassMSO}{\loc\mathbb{MSO}}
\newcommand{\ClassEMSO}{\loc\mathbb{EMSO}}
\newcommand{\ClassFO}{\loc\mathbb{FO}}

\newcommand{\extClassMSO}{\mathbb{MSO}}
\newcommand{\extClassEMSO}{\mathbb{EMSO}}
\newcommand{\extClassFO}{\mathbb{FO}}

\newcommand{\ClassCRA}{\mathbb{CRA}}
\newcommand{\ClassRCRA}{\mathbb{CRA}^-}
\newcommand{\ClassCMA}{\mathbb{CMA}}

\newcommand{\Comm}{\Sigma_{\mathsf{dyn}}}

\newcommand{\CommSign}{\Sign_{\mathsf{dyn}}^2}
\newcommand{\SimSign}{\smash{\Sign_{+1,\sim}^1}}

\newcommand{\CaSign}{\Sign^{2}_{+1,\sim}}
\newcommand{\RaSign}{\Sign^{1}_{\textup{+1}}}

\newcommand{\lsucc}{\edgesucc}
\newcommand{\lcsucc}{\oplus 1}

\newcommand{\ldsuccp}[2]{#1 + 1 = #2}
\newcommand{\lcsuccp}[2]{#1 \oplus 1 = #2}

\newcommand{\dpred}{\mathsf{succ}_{-}}
\newcommand{\dsucc}{\mathsf{succ}_{+}}

\newcommand{\nextwp}[1]{\mathsf{next}_{#1}}
\newcommand{\nextp}[2]{\mathsf{next}_{#2}^{#1}}

\newcommand{\prev}[1]{\mathsf{prev}_{#1}}
\newcommand{\prevp}[2]{\mathsf{prev}_{#2}^{#1}}

\newcommand{\csucc}[1]{\mathsf{next}_{#1}}
\newcommand{\csuccp}[2]{\mathsf{next}_{#2}^{#1}}

\newcommand{\cpred}[1]{\mathsf{prev}_{#1}}
\newcommand{\cpredp}[2]{\mathsf{prev}_{#2}^{#1}}

\newcommand{\Test}{T}
\newcommand{\Write}{U}

\renewcommand{\epsilon}{\varepsilon}

\newcommand{\Proc}{\mathit{Proc}}

\newcommand{\dom}{\textup{dom}}

\newcommand{\df}{:=}

\newcommand{\States}{Q}
\newcommand{\state}{q}

\newcommand{\pone}{c}
\newcommand{\ptwo}{d}

\newcommand{\spawnsymb}{\mathsf{f}}
\newcommand{\startsymb}{\mathsf{n}}
\newcommand{\sendsymb}{!}
\newcommand{\recsymb}{?}

\newcommand{\msgsymb}{\mathsf{msg}}
\newcommand{\procsymb}{\mathsf{proc}}
\newcommand{\forksymb}{\mathsf{fork}}

\newcommand{\spawn}[2]{\spawnsymb(#1,#2)}
\newcommand{\start}[2]{\startsymb(#1,#2)}
\newcommand{\send}[2]{!(#1,#2)}
\newcommand{\rec}[2]{?(#1,#2)}

\newcommand{\rparam}{\bar{\pi}}
\newcommand{\grparam}{\bar{\pi}}
\newcommand{\pto}{\mathrel{\rightharpoonup}}
\newcommand{\tspawn}[3]{#1 \leftarrow \mathsf{spawn}(#2,#3)}
\newcommand{\Act}{\mathit{Act}}

\renewcommand{\phi}{\varphi}


\title{An automaton over data words that captures EMSO logic}

\author{Benedikt Bollig}

\institute{ LSV, ENS Cachan, CNRS \& INRIA, France\\
\email{bollig@lsv.ens-cachan.fr} }

\maketitle

\begin{abstract}
We develop a general framework for the specification and implementation of systems whose executions are words, or partial orders, over an infinite alphabet. As a model of an implementation, we introduce class register automata, a one-way automata model over words with multiple data values. Our model combines register automata and class memory automata. It has natural interpretations. In particular, it captures communicating automata with an unbounded number of processes, whose semantics can be described as a set of (dynamic) message sequence charts. On the specification side, we provide a local existential monadic second-order logic that does not impose any restriction on the number of variables. We study the realizability problem and show that every formula from that logic can be effectively, and in elementary time, translated into an equivalent class register automaton.
\end{abstract}

  \gasset{AHangle=30,AHlength=1.4,AHLength=1.6}

 \def\propertyname{Fact}

\section{Introduction}

A recent research stream, motivated by models from XML database theory,
considers \emph{data words}, i.e., strings over an infinite alphabet \cite{Bou-IPL2002,Neven2004,DL-tocl08,Bjorklund10,Segoufin06}. The alphabet is the cartesian product of a finite supply of \emph{labels} and an
infinite supply of \emph{data values}.
While labels may represent, e.g., an
XML tag or reveal the type of an action that a system performs, data values
can be used to model time stamps
\iflong\cite{Bou-IPL2002,Bouyer2003,FHL-express2010}\fi
, process identifiers \cite{BH-csr10,Tzevelekos2011}, or text contents in XML documents\iflong~\cite{BMSS-jacm09}\fi.


We will consider data words as behavioral models of concurrent systems. In this regard, it is natural to look at suitable logics and automata. Logical formulas may serve as specifications, and automata as system models or tools for deciding logical theories. This viewpoint raises the following classical problems/tasks: \emph{satisfiability} (does a given logical formula have a model\,?), \emph{model checking} (do all executions of an automaton satisfy a given formula\,?), and \emph{realizability} (given a formula, construct a system model in terms of an  automaton whose executions are precisely the models of the formula). Much work has indeed gone into defining logics and automata for data words, with a focus on  satisfiability \cite{Bojanczy06,David2010}.

One of the first logical approaches to data words is due to \cite{Bou-IPL2002}.
Since then, a two-variable logic has become a commonly accepted yardstick wrt.\ expressivity and decidability \cite{Bojanczy06}. The logic contains a predicate to compare data values of two positions for equality. Its satisfiability problem is decidable, indeed, but supposedly of very high complexity. An elementary upper bound has been obtained only for weaker fragments \cite{Bojanczy06,David2010}. For specification of communicating systems, however, two-variable logic is of limited use: it cannot express properties like ``whenever a process Pid1 spawns some Pid2, then this is followed by a message from Pid2 to Pid1''. Actually, the logic was studied for words with only one data value at each each position, which is not enough to encode executions of message-passing systems. But three-variable logics as well as extensions to two data values lead to undecidability. To put it bluntly, any ``interesting'' logic for dynamic communicating systems has an undecidable satisfiability problem.

Instead of satisfiability or model checking, we therefore consider realizability. A system model that \emph{realizes} a given formula can be considered correct by construction. Realizability questions for data words have, so far, been neglected. One reason may be that there is actually no automaton that could serve as a realistic system model. Though data words naturally reflect executions of systems with an unbounded number of threads, existing automata fail to model distributed computation. Three features are minimum requirements for a suitable system model. First, the automaton should be a \emph{one-way device}, i.e., read an execution once, processing it ``from left to right'' (unlike data automata \cite{Bojanczy06}, class automata \cite{BL2010}, two-way register automata, and pebble automata \cite{Neven2004}). Second, it should be \emph{non-deterministic} (unlike alternating automata \cite{Neven2004,DL-tocl08}). Third, it should reflect paradigms that are used in concurrent programming languages such as process creation and message passing. Two known models match the first two properties: register automata \cite{Kaminski1994,KaminskiZ10,Tzevelekos2011} and class memory automata \cite{Bjorklund10}; but they clearly do not fulfill the last requirement.

\paragraph{Contribution.}

We provide an existential MSO logic over data words, denoted \EMSO, which does not impose any restriction on the number of variables. The logic is strictly more expressive than the two-variable logic from \cite{Bojanczy06} and suitable to express interesting properties of dynamic communicating systems.

We then define \emph{class register automata} as a system model. They are a mix of register automata \cite{Kaminski1994,KaminskiZ10,Tzevelekos2011} and class memory automata \cite{Bjorklund10}. A \myautomaton is a non-deterministic one-way device. Like a class memory automaton, it can access certain configurations in the past. However, we extend the notion of a configuration, which is no longer a simple state but composed of a state \emph{and} some data values that are stored in registers. This is common in concurrent programming languages and can be interpreted as ``read current state of a process'' or ``send process identity from one to another process''. Moreover, it is in the spirit of
communicating finite-state machines \cite{Brand1983} or nested-word automata \cite{Alur2009}, where more than one resource (state, channel, stack, etc.) can be accessed at a time. Actually, our automata run over directed acyclic graphs rather than words. To our knowledge, they are the first automata model of true concurrency that deals with structures over infinite alphabets.

We study the realizability problem and show that, for every \EMSO~formula, we can compute, in elementary time, an equivalent \myautomaton.
The effective translation is based on Hanf's locality theorem \cite{Hanf1965} and
\iflong
properly generalizes \cite{Bollig06,BK-IC08} to a dynamic setting with unbounded process creation.
\fi

\paragraph{Outline.}
Sections~\ref{sec:definitions} and \ref{sec:logic} introduce data words and their logics.
In Section~\ref{sec:CRA}, we define the new automata model.
Section~\ref{sec:autvslogic} is devoted to the realizability problem and states our main result. In Section~\ref{sec:specialcases}, we give translations from automata back to logic.
\iflong
An extension of our main result to infinite data words is discussed in Section~\ref{sec:infinite}.
\fi
We conclude in Section~\ref{sec:conclusion}.

\section{Data Words}\label{sec:definitions}

Let $\N = \{0,1,2,\ldots\}$ denote the set of natural numbers. For $\param
\in \N$, we denote by $[\param]$ the set $\{1,\ldots,\param\}$. A
\emph{boolean formula} over a (possibly infinite) set $A$ of \emph{atoms} is a finite object 
generated by the grammar $\beta ::= \mytrue \mid \myfalse \mid a \in A \mid \neg \beta \mid \beta \vee \beta \mid \beta \wedge \beta$. For an assignment of truth values to elements of $A$, a boolean formula $\beta$ is evaluated to true or false as usual. Its size $|\beta|$ is the number of vertices of its syntax tree. Moreover, $|A| \in \N \cup \{\infty\}$ denotes the size of a set $A$. The symbol $\,\mathord{\cong}\,$ will be used to denote isomorphism of two structures. For a partial function $f$, the domain of $f$ is denoted by $\dom(f)$.

\medskip

We fix an infinite set $\Data$ of \emph{data values}. Note that $\Data$ can be \emph{any} infinite set. For examples, however, we usually choose $\Data = \N$. In a data word, every position will carry $\param \ge 0$
data values. It will also carry a \emph{label} from a non-empty finite
alphabet $\Sigma$. Thus, a \emph{data word} is a finite sequence over $\Sigma
\times \Data^\param$ (over $\Sigma$ if $m=0$). Given a data word $w = (a_1,d_1) \ldots (a_n,d_n)$ with $a_i \in \Sigma$ and $d_i
= (d_i^1,\ldots,d_i^\param) \in \Data^m$, we let $\letter{i}$ refer to label $a_i$ and
$\dvalue{i}{k}$ to data value $d_i^k$.

Classical words without data come with natural relations on word positions such as the direct successor relation $\edge{+1}$ and its transitive closure $\mathord{<}$. In the context of data
words with one data value (i.e., $\param=1$), it is natural to consider also a
relation $\edge{\sim}$ for successive positions with identical data values
\cite{Bojanczy06}. As, in the present paper, we deal with multiple data
values, we generalize these notions in terms of a signature. A \emph{signature} $\Sign$ is a pair $(\RelSymb,\RelInt)$. It consists of a finite set $\RelSymb$ of binary \emph{relation symbols} and an \emph{interpretation} $\RelInt$. The latter associates, with every $\rel \in \RelSymb$ and every data word $w = w_1 \ldots w_n \in (\Sigma \times \Data^\param)^\ast$, a relation $\Int{\rel}{w} \subseteq [n] \times [n]$ such that the following hold, for all word positions $i,j,i',j' \in [n]$:
\begin{itemize}
\item[(1)] $i \Int{\rel}{w} j$ implies $i < j$
\item[(2)] there is at most one $k$ such that $ i \Int{\rel}{w} k$
\item[(3)] there is at most one $k$ such that $k \Int{\rel}{w} i$
\item[(4)] if $i \Int{\rel}{w} j$ and $i' \Int{\rel}{w} j'$ and $w_i = w_{i'}$ and $w_j = w_{j'}$, then $i < i'$ iff $j < j'$
\end{itemize}
In other words, we require that $\Int{\rel}{w}$ (1) complies with $<$, (2) has out-degree at most one, (3) has in-degree at most one, and (4) is monotone. Our translation from logic into automata will be symbolic and independent of $\RelInt$, but its applicability and correctness rely upon the above conditions. However, several examples will demonstrate that the framework is quite flexible and allows us to capture existing logics and automata for data words. Note that $\Int{\rel}{w}$ can indeed be \emph{any} relation satisfying (1)--(4). It could even assume an order on $\Data$.

As the interpretation $\RelInt$ is mostly understood, we may identify $\Sign$ with $\RelSymb$ and write $\rel \in \Sign$ instead of $\rel \in \RelSymb$, or $|\Sign|$ to denote $|\sign|$. If not stated otherwise, we let in the following $\Sign$ be any signature.

\begin{example}\label{ex:sim}
Typical examples of relation symbols include $\edge{+1}$ and $\edge{\tequiv}^k$ relating direct successors and, respectively, successive positions with the same $k$-th data value: For $w = w_1 \ldots w_n$, we let $\mathord{\edge{+1}^w} = \{(i,i+1) \mid i \in \{1,\ldots,n-1\}\}$ and $(\mathord{\edge{\tequiv}^k)^w} = \{(i,j) \mid 1 \le i < j \le n$,  $\dvalue{i}{k} = \dvalue{j}{k}$, and there is no $i < i' < j$ such that $\dvalue{i}{k} = \dvalue{i'}{k}\}$. When $m = 1$, we write $\edge{\sim}$ instead of $\edge{\sim}^1$.
Automata and logic have been well studied in the presence of
one single data value ($m = 1$) and for signature $\SimSign =
\{\edge{+1}\,,\,\edge{\sim}\}$ with the above interpretation \cite{Bojanczy06,Bjorklund10}. Here, and in the following, we adopt the convention that the upper index of a signature denotes the number $m$ of data values.
Figure~\ref{fig:dataword} depicts a data word over $\Sigma = \{\req,\ack\}$ (request/acknowledgment) and $\Data = \N$ as well as the relations $\edge{+1}$ (straight arrows) and $\edge{\sim}$ (curved arrows) imposed by $\SimSign$. \eoex
\end{example}

\setlength\belowcaptionskip{-0.2cm}

\begin{figure}[t]
\centering
{
\begin{minipage}[b]{0.44\textwidth}
\centering
{
 \scalebox{0.85}{
\begin{picture}(56,16)(-1,-4)
\newcommand{\dist}{\hspace{0.8em}&\hspace{0.8em}}
\gasset{Nh=4,Nw=3.8,Nadjustdist=1,AHangle=35,AHLength=1.0,AHlength=0.4,Nframe=n,Nfill=n,linewidth=0.1}
\unitlength=0.24em
\put(0,-6.5){
  \node(w)(35,0){$
\begin{array}{ccccccccc}
  8 \dist 5 \dist 3 \dist 4 \dist 3 \dist 4 \dist 5 \dist 4\\~
\end{array}
$}
}
\node[Nh=3](A1)(0,0){$\req$}
\node[Nh=3](A2)(10,0){$\req$}
\node[Nh=3](A3)(20,0){$\req$}
\node[Nh=3](A4)(30,0){$\req$}
\node(A5)(40,0){}
\node(A6)(50,0){}
\node(A7)(60,0){}
\node(A8)(70,0){}

\node(b)(40,0.4){$\ack$}
\node(b)(50,0.4){$\ack$}
\node(b)(60,0.4){$\ack$}
\node(b)(70,0.4){$\ack$}

\drawedge[curvedepth=7](A3,A5){}
\drawedge[curvedepth=7](A4,A6){}
\drawedge[curvedepth=7](A6,A8){}
\drawbcedge(A2,16,16,A7,54,16){$\edgesim$}

\drawedge[ELside=l,ELpos=50](A1,A2){\small$\edgesucc$}
\drawedge(A2,A3){}
\drawedge(A3,A4){}
\drawedge(A4,A5){}
\drawedge(A5,A6){}
\drawedge(A6,A7){}
\drawedge(A7,A8){}
\end{picture}
}
}
\caption{Data word over $\SimSign$\label{fig:dataword}}\vspace{0.4ex}
\end{minipage}
}
~
{
\begin{minipage}[b]{0.48\textwidth}
\centering
{
\scalebox{0.84}{
\begin{picture}(63,27)(0,-20)
\newcommand{\dist}{\hspace{0.6em}&\hspace{0.6em}}
\gasset{Nh=4,Nw=3.8,Nadjustdist=1,AHangle=35,AHLength=1.0,AHlength=0.4,Nframe=n,Nfill=n,linewidth=0.1}
\unitlength=0.20em
\put(10,-20){
 \node(w)(40,-5.9){$
\begin{array}{ccccccccccc}
  \startsymb \dist \spawnsymb \dist \startsymb \dist \spawnsymb \dist \startsymb \dist \sendsymb \dist \recsymb \dist \sendsymb \dist \sendsymb \dist \recsymb \dist \recsymb\\
  2 \dist 2 \dist 3 \dist 2 \dist 1 \dist 2 \dist 3 \dist 1 \dist 1 \dist 3 \dist 3\\
  2 \dist 3 \dist 2 \dist 1 \dist 2 \dist 3 \dist 2 \dist 3 \dist 3 \dist 1 \dist 1
\end{array}
$}
}
%
%
\node(A1)(0,0){}
\node(A2)(10,0){}
 \node(A3)(20,-10){}
 \node(A4)(30,0){}
 \node(A5)(50,0){}
 \node(A6)(40,10){}
 \node(A7)(60,-10){}
 \node(A8)(70,10){}
 \node(A9)(80,10){}
 \node(A10)(90,-10){}
 \node(A11)(100,-10){}
 %
 %
 \node(b)(0,-0.4){$\startsymb$}
 \node(b)(10,0){$\spawnsymb$}
 \node(b)(20,-10.4){$\startsymb$}
 \node(b)(30,0){$\spawnsymb$}
 \node(b)(50,0){$\sendsymb$}
 \node(b)(40,9){$\startsymb$}
 \node(b)(60,-10){$\recsymb$}
 \node(b)(70,10){$\sendsymb$}
 \node(b)(80,10){$\sendsymb$}
 \node(b)(90,-10){$\recsymb$}
 \node(b)(100,-10){$\recsymb$}

 \drawedge(A1,A2){}
 \drawedge[ELside=r,ELdist=0.05,ELpos=30](A2,A3){$\edge{\forksymb}$}
 \drawedge[ELside=l,ELdist=0.3](A2,A4){$\edge{\procsymb}$}
 \drawedge(A4,A6){}
 \drawedge(A5,A7){}
 \drawedge[ELside=r,ELpos=30,ELdist=0.1](A8,A10){$\fedge{\msgsymb}$}
 \drawedge[ELside=l,ELpos=60,ELdist=0.1](A9,A11){$\fedge{\msgsymb}$}
 \drawedge(A3,A7){}
 \drawedge(A4,A5){}
 
  \drawedge(A6,A8){}
  \drawedge(A7,A10){}
  \drawedge(A8,A9){}
  \drawedge(A10,A11){}
 \end{picture}
}
}
\caption{Data word over $\CommSign$\label{fig:msc}}
\vspace{0.1ex}
\end{minipage}
}
\end{figure}

\begin{example}\label{ex:comm}
  We develop a framework for message-passing systems with dynamic process creation. Each process has a unique identifier from $\Data = \N$. Process $\pone \in \N$
  can execute an action $\spawn{\pone}{\ptwo}$, which forks a new process with
  identity $\ptwo$. This action is eventually followed by
  $\start{\ptwo}{\pone}$, indicating that $\ptwo$ is new (created by $\pone$)
  and begins its execution. Processes can exchange messages. When $c$ executes $\send{\pone}{\ptwo}$, it sends a message through an unbounded first-in-first-out (FIFO) channel \mbox{$c \to d$}. Process $d$ may execute $\rec{\ptwo}{\pone}$ to receive the message.
  Elements from $\Comm =
  \{\spawnsymb,\startsymb\,,\,\sendsymb\,,\recsymb\}$ reveal the nature of
  an action, which requires two identities so that we choose $\param
  = 2$. When a process performs an action, it should access
  the current state of (i) its own, (ii) the
  spawning process if a new-action is executed, and (iii) the sending process if a receive is
  executed (message contents are encoded in states).
  To this aim, we define a signature $\CommSign = \{\edge{\procsymb}\,,\,\edge{\forksymb}\,,\,\fedge{\msgsymb}\}$ with the following interpretation. Assume $w = w_1 \ldots w_n \in (\Comm \times \N \times \N)^\ast$ and consider, for $a,b \in \Comm$ and $i,j \in [n]$, the property
\[P_{(a,b)}(i,j) ~=~ (\letter{i} = a
  \mathrel{\wedge} \letter{j} = b \mathrel{\wedge} \dvalue{i}{1} =
  \dvalue{j}{2} \mathrel{\wedge} \dvalue{i}{2} = \dvalue{j}{1})\,.\]
We set $\mathord{\edge{\procsymb}^w} = (\mathord{\edge{\sim}^1})^w$, which relates successive positions with the same executing process. Moreover, let 
$i \mathrel{\edge{\forksymb}^w} j$ if $i < j$, $P_{(\spawnsymb,\startsymb)}(i,j)$, and there is no $i < k < j$ such that $P_{(\spawnsymb,\startsymb)}(i,k)$ or $P_{(\spawnsymb,\startsymb)}(k,j)$.
Finally, we set
$i \mathrel{\edge{\msgsymb}^w} j$ if $i < j$, $P_{(!,?)}(i,j)$, and
\[|\{i'  < i \mid P_{(!,?)}(i',j)\,\}| ~=~ |\{j'  < j \mid P_{(!,?)}(i,j')\,\}|\,.\]
This models FIFO communication.
An example data word is given in Figure~\ref{fig:msc}, which also depicts the relations induced by $\CommSign$. Horizontal arrows reflect $\edge{\procsymb}$, vertical arrows either $\edge{\forksymb}$ or $\edge{\msgsymb}$, depending on the labels. Note that $\start{2}{2}$ is executed by ``root process'' $2$, which was not spawned by some other process.
\eoex
\end{example}


\newcommand{\mywrap}{
\begin{wrapfigure}[2]{r}{0.17\textwidth}
\vspace{-23pt}
\centering
{
\scalebox{0.85}{
\begin{picture}(20.5,7)(20,-0.5)
\gasset{Nh=4,Nw=3.7,Nadjustdist=1,AHangle=35,AHLength=1.0,AHlength=0.4,Nframe=n,Nfill=n,linewidth=0.1}
\unitlength=0.24em
\node[Nh=3](A3)(20,0){$\req$}
\node[Nw=3.5,Nmr=0,Nframe=y,Nh=4](A4)(30,0){$\req$}
\node(A5)(40,0){}
\node(A6)(50,0){}

\node(b)(40,0.4){$\ack$}
\node(b)(50,0.4){$\ack$}

\drawedge[curvedepth=7](A3,A5){}
\drawedge[curvedepth=7](A4,A6){}

\drawedge(A3,A4){}
\drawedge(A4,A5){}
\drawedge(A5,A6){}
\end{picture}
}
}
\end{wrapfigure}
}

\iflong
\paragraph{Graph Abstraction.}
Note that the graph induced by the data word from Figure~\ref{fig:msc} does not resemble a word anymore, as the direct successor relation on word positions is abandoned.
Actually, we can see data words from a different angle. A signature $\Sign$ determines a class of \emph{data graphs} $\mathcal{G}$ with $(\Sigma \times \Data^\param)$-labeled nodes and $\Sign$-labeled edges. A data graph is contained in $\mathcal{G}$ if it can
  be ``squeezed'' into a word $w$ such that nodes that are connected by a
  $\rel$-labeled edge turn into word positions that are related by
  $\rel^w$. In other words, we consider directed acyclic graphs such
  that at least one linearization (extension to a total order) matches the
  requirements imposed by the signature.
  
  Our principal proof technique relies on a graph abstraction of data words
where data values are classified into equivalence classes.
\fi
Let $\Part{m}$ be the set of all partitions of $[\param]$. An $\Sign$-\emph{graph} is a (node- and edge-labeled) graph $\graph = (V,(\rel^\graph)_{\rel \in \Sign},\labelone,\labeltwo)$. Here, $V$ is the finite set of nodes, $\labelone: V \to \Sigma$ and $\labeltwo: V \to \Part{m}$ are node-labeling functions, and each $\rel^\graph \subseteq V \times V$ is a set of edges such that, for all $i \in V$, there is at most one $j \in V$ with $i \rel^\graph j$, and there is at most one $j \in V$ with $j \rel^\graph i$. We represent $\rel^\graph$ and $(\rel^\graph)^{-1}$ as partial functions and set $\csuccp{\graph}{\rel}(i) = j$ if $i \rel^\graph j$, and $\cpredp{\graph}{\rel}(i) = j$ if $j \mathrel{\rel^\graph} i$.

Local graph patterns, so-called spheres, will also play a key role. For nodes $i,j \in V$, we denote by $\mydist^\graph(i,j)$ the \emph{distance} between $i$ and $j$, i.e., the length of the shortest path from $i$ to $j$ in the undirected graph $(V \,,\,\bigcup_{\rel \in \Sign} \mathord{\rel^\graph} \,\cup\, \mathord{(\rel^\graph)^{-1}})$ (if such a path exists). In particular, $\smash{\mydist^\graph}(i,i) = 0$. For some \emph{radius} $\radius \in \N$, the $\radius$\emph{-sphere of} $\graph$ \emph{around} $i$, denoted by $\Sph{B}{\graph}{i}$, is the substructure of $\graph$ induced by $\{j \in V \mid \mydist^\graph(i,j) \le B\}$. In addition, it contains the distinguished element $i$ as a constant, called \emph{sphere center}.

These notions naturally transfer to data words: With word $w$ of length $n$, we associate the graph $\Graph{w} = ([n],(\rel^w)_{\rel \in \Sign},\labelone,\labeltwo)$ where $\labelone$ maps $i$ to $\letter{i}$ and $\labeltwo$ maps $i$ to $\{\{l \in [\param] \mid \dvalue{i}{k} = \dvalue{i}{l}\} \mid k \in [\param]\}$. Thus, $K \in \labeltwo(i)$ contains indices with the same data value at position $i$. Now, $\csucc{\rel}^w$, $\cpred{\rel}^w$, $\mydist^w$, and $\Sph{B}{w}{i}$ are defined with reference to the graph $\Graph{w}$. We hereby assume that $\Sign$ is understood. We might also omit the index $w$ if it is clear from the context.

Data words $u$ and $v$ are called ($\Sign$-)\emph{equivalent} if $\Graph{u} \cong \Graph{v}$. For a language $L$, we let $[L]_{\Sign}$ denote the set of words that are equivalent to some word in $L$.

Given the data word $w$ from Figure~\ref{fig:dataword}, we have $\mydist^w(1,8) = 3$. The picture\mywrap on the right shows $\Sph{1}{w}{4}$. The sphere center is framed by a  rectangle; node labelings  of the form $\{\{1\}\}$ are omitted.

\section{Logic}\label{sec:logic}

We consider monadic second-order logic to specify properties of data words.
Let us fix countably infinite supplies of first-order variables $\{x,y,\ldots\}$ and second-order variables $\{X,Y,\ldots\}$.

\medskip

The set $\extMSO(\Sign)$ of \emph{monadic second-order formulas} is given by the grammar
\[\phi \,::=\,  \lformula{x}{a} \mid \dvalue{x}{k} = \dvalue{y}{l}  \mid  x \rel y
\mid x = y \mid x \in X \mid \neg \phi \mid \phi \mathrel{\vee} \phi \mid
\exists x\, \phi \mid \exists X\, \phi\]
where $a \in \Sigma$, $k,l \in [\param]$, $\rel \in \Sign$, $x$ and $y$ are
first-order variables, and $X$ is a second-order variable. The \emph{size} $|\phi|$ of $\phi$ is the number of nodes of its syntax tree.

Important fragments of $\extMSO(\Sign)$ are $\extFO(\Sign)$, the set of first-order formulas, which do not use any second-order quantifier, and $\extEMSO(\Sign)$, the set of formulas of the form $\exists X_1 \ldots \exists X_n\, \phi$ with $\phi \in \extFO(\Sign)$.

The models of a formula are data words. First-order variables are
interpreted as word positions and second-order variables as sets of positions.
Formula $\lformula{x}{a}$ holds in data word $w$ if position $x$ carries an $a$, and
formula $\dvalue{x}{k} = \dvalue{y}{l}$ holds if the $k$-th data value at position $x$ equals the $l$-th data value at position $y$. Moreover, $x \mathrel{\rel} y$ is satisfied if $x \mathrel{\rel^w} y$. The atomic formulas $x = y$ and $x \in X$ as
well as quantification and boolean connectives are interpreted as usual.

For realizability, we will actually consider a restricted, more ``local'' logic: let $\MSO(\Sign)$ denote the fragment of $\extMSO(\Sign)$ where we can only use $\dvalue{x}{k} = \dvalue{x}{l}$ instead of the more general $\dvalue{x}{k} = \dvalue{y}{l}$. Thus, data values of \emph{distinct} positions can only be compared via $x \mathrel{\rel} y$. This implies that $\MSO(\Sign)$ cannot distinguish between words $u$ and $v$ such that $\Graph{u} \cong \Graph{v}$. The fragments $\FO(\Sign)$ and $\EMSO(\Sign)$ of $\MSO(\Sign)$ are defined as expected.

In the case of one data value ($\param=1$), we will also refer to the
logic \mbox{$\EMSOtwo$} that was considered in \cite{Bojanczy06}
and restricts EMSO logic to two first-order variables. The predicate $<$ is interpreted as the strict linear order on word positions (strictly speaking, it is not part of a signature as we defined it). We shall later see that $\EMSO(\SimSign)$ is strictly more expressive than \mbox{$\EMSOtwo$}, though the latter involves the non-local predicates $\dvalue{x}{1} = \dvalue{y}{1}$ and $<$. This gain in expressiveness comes at the price of an undecidable satisfiability problem.

A \emph{sentence} is a formula without free variables. The language defined by
sentence $\phi$, i.e., the set of its models, is denoted by $L(\phi)$. By
$\extClassMSO(\Sign)$, $\ClassMSO(\Sign)$, $\ClassEMSO(\Sign)$, etc., we refer to the corresponding language classes.

\begin{example}\label{ex:server}
Think of a server that can receive requests ($\req$) from an unbounded number of processes, and acknowledge ($\ack$) them. We let $\Sigma = \{\req,\ack\}$, $\Data = \N$, and $m=1$. A data value from $\Data$ is used to model the process identity of the requesting and acknowledged process. We present three properties formulated in $\FO(\SimSign)$. Formula $\phi_1 = \exists x\exists y\,(\lformula{x}{\req} \,\wedge\, \lformula{y}{\ack} \,\wedge\, x \edge{\sim} y)$ expresses that there is a request that is acknowledged. Dually, $\phi_2 = \forall x \exists y\,(\lformula{x}{\req} \,\rightarrow\, \lformula{y}{\ack} \mathrel{\wedge} x \edge{\sim} y)$ says that every request is acknowledged before the same process sends another request. A last formula guarantees that two \emph{successive} requests are acknowledged in the order they were received:
\[\phi_3 = \forall x,y\,
\left(
\begin{array}{cl}
& \lformula{x}{\req} \,\wedge\, \lformula{y}{\req} \,\wedge\, x \edgesucc y\\
\rightarrow \exists x',y'\,
      \bigl( & \lformula{x'}{\ack} \,\wedge\, \lformula{y'}{\ack}
        \,\wedge\, x \edgesim x' \edgesucc y' \wedge y \edgesim y' \,\bigr)
\end{array}
\right)\]
This is not expressible in \mbox{$\EMSOtwo$}. We will see that $\phi_1,\phi_2,\phi_3$ form a hierarchy of languages that correspond to different automata models, our new model capturing $\phi_3$. \eoex
\end{example}

\begin{example}
\newcommand{\Prop}{\phi}
We pursue Example~\ref{ex:comm} and consider $\Comm$ with signature $\CommSign$. Recall that we wish to model systems where an unbounded number of processes communicate via message-passing through unbounded FIFO channels. 
Obviously, not every data word represents an execution of such a system. Therefore, we identify some \emph{well formed} data words,
  which have to satisfy $\Prop_1 \mathrel{\wedge} \Prop_2 \mathrel{\wedge} \Prop_3 \in \FO(\CommSign)$ given as follows. We require that there is exactly one root
  process:
$\Prop_1 = \exists x\, \left(
 \lformula{x}{\startsymb} \,\wedge\,\dvalue{x}{1} = \dvalue{x}{2}
        \,\wedge\, \forall y\,( \dvalue{y}{1} = \dvalue{y}{2} \,\rightarrow\, x = y )
       \right)$.
Next, we assume that every fork is followed by a corresponding new-action, the first action of a process is a new-event, and every new process was forked by some other process:
\[\Prop_2 = \forall x
  \left(
      \begin{array}{ll}
        & \lformula{x}{\,\spawnsymb} \,\rightarrow\, \exists y\, (x \edge{\forksymb} y)\\
        \wedge & \lformula{x}{\,\startsymb} \,\leftrightarrow\, \neg\exists y\, (y \edge{\procsymb} x)\\
        \wedge & \lformula{x}{\,\startsymb} \,\rightarrow\, \left(\dvalue{x}{1} = \dvalue{x}{2} \vee \exists y\, (y \edge{\forksymb} x)\right)
  \end{array}\right)\]
Finally, every send should be followed by a
receive, and a receive be preceded by a send action:
$\Prop_3 = \forall x\, \bigl(\ell(x) \in \{\,\sendsymb\,,\recsymb\,\} \rightarrow \exists y\,(x \fedge{\msgsymb}
  y \,\vee\, y \fedge{\msgsymb} x)\bigr)$.
This formula actually ensures that, for every $c,d \in \N$,
there are as many symbols $\send{c}{d}$ as $\rec{d}{c}$, the $N$-th send symbol being matched with the $N$-th receive symbol. We call a data word over $\Comm$ and $\CommSign$ a \emph{message sequence chart} (MSC, for short) if it satisfies $\Prop_1 \mathrel{\wedge} \Prop_2 \mathrel{\wedge} \Prop_3$. Figure~\ref{fig:msc} shows an MSC and the induced relations. When we restrict to MSCs, our logic corresponds to that from \cite{LeuckerMM02}. Note that model checking $\MSO(\CommSign)$ specifications against \emph{fork-and-join grammars}, which can generate infinite sets of MSCs, is decidable \cite{LeuckerMM02}.

A last $\FO(\CommSign)$-formula (which is not satisfied by all MSCs) specifies that, whenever a process $c$ forks some $d$, then this is followed by a message from $d$ to $c$:
$\forall x_1,y_1\,
\left(
 x_1 \edge{\forksymb} y_1 \,\rightarrow\, \exists x_2,y_2\, (x_1 \edge{\procsymb} x_2 \,\wedge\, y_1 \edge{\procsymb} y_2 \fedge{\msgsymb} x_2 )\right)$.
\eoex
\end{example}

\section{Class Register Automata}\label{sec:CRA}

In this section, we define \myautomata, a non-deterministic one-way automata model that captures \EMSO~logic. It combines register automata \cite{Kaminski1994,KaminskiZ10} and class memory automata \cite{Bjorklund10}.
When processing a data word, data values from the current position can be stored
in registers. The automaton reads the data word from left to right but can look
back on certain states and register contents from the past (e.g., at the last position that is executed by the same process).
Positions that can be accessed in this way are determined by the signature $\Sign$.
Their register entries can be compared with one another, or with current values from the input. Moreover, when taking a transition, registers can be updated by 
either a current value, an old register entry, or a guessed value.

\begin{definition}
  A \emphmyautomaton (over signature $\Sign$) is a tuple
  $\Aut = (\States,\Reg,\Trans,{(\Local_\rel)}_{\rel \in \Sign},\Acc)$ where
\begin{itemize}\itemsep=0.5ex
\item
$\States$ is a finite set of \emph{states},
\item
$\Reg$ is a finite set of \emph{registers},
\item
the $\Local_\rel \subseteq \States$ are sets of \emph{local final states},
\item
$\Acc$ is the \emph{global acceptance condition:} a boolean formula over $\{\,\textup{`}q
\le N\textup{'} \mid q \in Q$ and $N \in \N\}$, and
\item
$\Trans$ is a finite set of \emph{transitions} of the form
  \[\trans{(\source,\test)}{a}{(\state,\update)}\,.\]
  Here, $\source: \Sign \pto Q$ is a partial mapping representing the source states. Moreover,  $\test$ is a guard, i.e., a boolean formula over $\{\,\textup{`}\theta_1 = \theta_2\textup{'} \mid \theta_1,\theta_2 \in [m] \cup (\dom(\source) \mathrel{\times} \Reg)\}$ to perform comparisons of values that are are currently read and those that are stored in registers. Finally, $a \in \Sigma$ is the current label, $q \in Q$ is the target state, and
  $\update: R \pto (\dom(\source) \times \Reg) \mathrel{\cup} ([m] \times \N)$ is a partial mapping to update registers.
\end{itemize}
\end{definition}

In the following, we write $\source_\rel$ instead of $\source(\rel)$.
Transition $\trans{(\source,\test)}{a}{(\state,\update)}$ can be
executed at position $i$ of a data word if the state at position
$\prev{\rel}(i)$ is $\source_\rel$ (for all $\rel \in \dom(\source)$) and, for a register
guard $(\rel_1,r_1) = (\rel_2,r_2)$, the entry of register $r_1$
at $\prev{\rel_1}(i)$ equals that of $r_2$ at $\prev{\rel_2}(i)$. The
automaton then reads the label $a$ together with a tuple of data values that
also passes the test given by $\test$, and goes to $q$. Moreover, register $r$ obtains a new
value according to $\update(r)$: if $\update(r) = (\rel,r') \in \dom(\source) \times \Reg$, then the new value of $r$ is the value of $r'$ at position $\prev{\rel}(i)$; if $\update(r) = (k,B) \in [m] \times \N$, then $r$ obtains any $k$-th data value in the $B$-sphere around $i$. In particular, $\update(r) = (k,0)$ assigns to $r$ the (unique) $k$-th data value of the current position. To some extent, $\update(r) = (k,B)$ calls an oracle to guess a data value. The guess is local and, therefore, weaker than \cite{KaminskiZ10}, where a non-deterministic reassignment allows one to write \emph{any} data value into a register. This latter approach can indeed simulate our local version (this is not immediately clear, but can be shown using the \emph{sphere automaton} from Section~\ref{sec:autvslogic}).

Let us be more precise. A configuration of $\Aut$ is a pair $(\state,\reg)$ where
$\state \in Q$ is the current state and $\reg: \Reg \pto \Data$ is a partial mapping denoting the current register contents. If $\reg(r)$ is undefined, then there is no entry in $r$. Let $w = w_1 \ldots w_n \in (\Sigma \times \Data^m)^\ast$ be a data word and $\xi = (\state_1,\reg_1) \ldots (\state_n,\reg_n)$ be a sequence of configurations. For $i \in [n]$, $k \in [m]$, and $B \in \N$,
let $\Data_B^k(i) = \{\dvalue{j}{k} \mid j \in [n]$ such that $\mydist^w(i,j) \le B\}$. We call $\xi$ a \emph{run} of $\Aut$ on $w$ if,
for every position $i \in [n]$, there is a transition
$\trans{(\source_i,\test_i)}{\letter{i}}{(\state_i,\update_i)}$ such that the following hold:

\begin{itemize}\itemsep=1ex
\item[(1)] $\dom(\source_i) = \{\rel \in \Sign \mid \prev{\rel}(i)$ is defined$\}$

\item[(2)] for all $\rel \in \dom(\source_i)$\,: ${(\source_i})_\rel = \state_{\prev{\rel}(i)}$

\item[(3)] $\test_i$ is evaluated to true on the basis of its atomic subformulas: $\theta_1 = \theta_2$ is true iff $\mathit{val}_i(\theta_1) = \mathit{val}_i(\theta_2) \in \Data$
where $\mathit{val}_i(k) = \dvalue{i}{k}$ and $\mathit{val}_i((\rel,r)) = \reg_{\prev{\rel}(i)}(r)$ (the latter might be undefined and, therefore, not be in $\Data$)

\item[(4)] for all $r \in \Reg$\,: $\left\{\begin{array}{ll}
      \reg_i(r) = \reg_{\prev{\rel}(i)}(r') &  \textup{if~} \update_i(r)=(\rel,r') \,\in \dom(\source) \times \Reg\\[0.5ex]
      \reg_i(r) \in \Data_B^k(i) & \textup{if~} \update_i(r) = (k,B) \,\in [m] \times \N\\[0.5ex]
      \reg_i(r)~\textup{undefined} & \textup{if~} \update_i(r) \textup{~undefined}
    \end{array}\right.$
\end{itemize}
\smallskip
Run $\xi$ is accepting if $q_i \in \Local_\rel$ for all $i \in [n]$ and $\rel \in \Sign$ such that $\nextwp{\rel}(i)$ is undefined. Moreover, we
require that the global condition $\Acc$ is met. Hereby, an atomic
constraint $q \le N$ is satisfied by $\xi$ if $|\{\,i \in [n] \mid
q_i = q\}| \le N$. The language $L(\Aut) \subseteq (\Sigma \times \Data^m)^\ast$ of $\Aut$ is defined in the obvious manner. The corresponding language class is denoted by
$\ClassCRA(\Sign)$.

\medskip

The acceptance conditions are inspired by Bj{\"o}rklund and Schwentick
\cite{Bjorklund10}, who also distinguish between local and global acceptance. Local final states can be motivated as follows. When data values
model process identities, a $\edge{\sim}$-maximal position of a data word is the last
position of some process and must give rise to a local final
state. Moreover, in the context of $\CommSign$, a sending position that does not lead to a local final state in $F_{\fedge{\msgsymb}}$ requires a matching receive event.
Thus, local final states can be used to model ``communication requests''. The global acceptance condition of \myautomata is more general than that of \cite{Bjorklund10} to cope with all possible signatures. However, in the special case of $\SimSign$, there is some global control in terms of $\edge{+1}$. We could then perform some counting up to a finite threshold and restrict, like \cite{Bjorklund10}, to a set of global final states.

We can classify many of the non-deterministic one-way models from the literature (most of them defined for $m=1$) in our unifying framework:


\begin{itemize}\itemsep=1ex
\item A \emphcmautomaton \cite{Bjorklund10} is a \myautomaton where, in all
  transitions $\smash{\trans{(\source,\test)}{a}{(\state,\update)}}$, the update function $\update$ is undefined everywhere.
The corresponding language class is denoted by $\ClassCMA(\Sign)$.

\item As an intermediary subclass of \myautomata, we consider \emphngautomata: for all transitions $\trans{(\source,\test)}{a}{(\state,\update)}$ and registers $r$, one requires $\update(r) \in (\dom(\source) \times \Reg) \mathrel{\cup} ([m] \times \{0\})$.
We denote the corresponding language class by $\ClassRCRA(\Sign)$.

\item A \emph{register automaton} \cite{Kaminski1994,DL-tocl08} is a \ngautomaton over $\Sign_{+1}^m = \{\edge{+1}\}$. Moreover, non-guessing \myautomata over $\SimSign$ capture \emph{fresh-register automata} \cite{Tzevelekos2011}, which can dynamically generate data values that do not occur in the history of a run. Actually, this feature is also present in dynamic communicating automata \cite{BH-csr10} and in \cmautomata over $\SimSign$ where a fresh data value is guaranteed by a transition $\trans{(\source,\test)}{a}{(\state,\update)}$ such that $\source_{\edge{\sim}}$ is undefined.

\item  Class register automata are a model of distributed computation: considered over $\Comm$ and $\CommSign$, they subsume dynamic communicating automata \cite{BH-csr10}. In particular, they can handle unbounded process creation and message passing. Updates of the form $\update(r) = (\edge{\forksymb},r')$ and $\update(r) = (\fedge{\msgsymb},r')$ correspond to receiving a process identity from the spawning/sending process. Moreover, when a process requests a message from the thread whose identity is stored in register $r$, a corresponding transition is guarded by $(\edge{\procsymb},r) = (\fedge{\msgsymb},r_0)$ where we assume that every process keeps its identity in some register $r_0$.
\end{itemize}


\setlength\belowcaptionskip{-0.2cm}

 \begin{example}\label{ex:rCRA}
 Let us give a concrete example. Suppose $\Sigma = \{\req,\ack\}$ and $\Data = \N$.  We pursue Example~\ref{ex:server} and build a \ngautomaton $\Aut$ over $\SimSign$ for
 $L = [\{(\req,1) \ldots (\req,n)(\ack,1) \ldots (\ack,n) \mid n \ge 1\}]_{{\SimSign}}$. 
Roughly speaking, there is a request phase followed by an acknowledgment phase, and requests are acknowledged in the order they are received. Figure~\ref{table:rCRA} presents $\Aut$ and an accepting run on $(\req,8)(\req,5)(\ack,8)(\ack,5)$. The states of $\Aut$ are $q_1$ and $q_2$. State $q_1$ is assigned to request positions (first phase), state $q_2$ to acknowledgments (second phase). Moreover, $\Aut$ is equipped with registers $r_1$ and $r_2$. During the first phase, $r_1$ always contains the data value of the current position, and $r_2$ the data value of the $\edge{+1}$-predecessor (unless we deal with the very first position, where $r_2$ is undefined, denoted $\bot$). These invariants are ensured by transitions 1 and 2. In the second phase, by transition 3, position $n+1$ carries the same data value as the first position, which is the only request with undefined $r_2$. Guard $(\edge{\sim},r_2) = \bot$ is actually an abbreviation for $\neg((\edge{\sim},r_2) = (\edge{\sim},r_2))$. By transition 4, position $n+i$ with $i \ge 2$ has to match the request position whose $r_2$-contents equals $r_1$ at $n+i-1$. Finally, $F_{\edge{\sim}} = \{q_2\}$, $F_{\edge{+1}} = \{q_2\}$, and $\Acc = \neg(q_1 \le 0)$.
\eoex
 \end{example}

\newcommand{\height}{\parbox[0pt][3.5ex][c]{0cm}{}}
\newcommand{\dheight}{\parbox[0pt][6ex][c]{0cm}{}}
 
  \begin{figure*}[t]
  \centering
 \scalebox{0.85}{
   $\begin{array}{|r|c|c|c|c|c|l|c|c|c|c|c|}
   \multicolumn{6}{c}{~~~~~~~~~~~~~~~~~~~\textup{Transitions}} & \multicolumn{6}{c}{~~~~~~~~~~~~~~~~~~~~~~~~~~~~~~~~~~\textup{Run}}\\
     \hhline{-------~----}
     & \multicolumn{2}{c|}{\textup{source (}\source\textup{)}} & \textup{guard (}\test\textup{)} & ~\textup{input}~ &  \parbox{2em}{\centering$q$} &
     \parbox{7em}{\centering update ($\update$)} & \height~~~~~~~~~~ & ~\textup{input}~ & \,\textup{state}\, & ~r_1~ & ~r_2~\\
     & \parbox{2em}{\centering$\edge{\sim}$} & \parbox{2em}{\centering$\edge{+1}$} & & & & & & & & &\\
     \hhline{-------~----}
     ~1~ & \height & & & (\req,d) & q_1 & ~\begin{array}{l} r_1 := d \end{array} & & (\req,8) & q_1 & 8 & \bot\\
     \hhline{-------~----}
     ~2~ & \dheight & q_1 &  & (\req,d) & q_1 & ~\begin{array}{l} r_1 := d\\r_2 := (\edge{+1},r_1) \end{array} & & (\req,5) & q_1 & 5 & 8\\
     \hhline{-------~----}
     ~3~ & \height q_1 & q_1 & ~ (\edge{\sim},r_2) = \bot~ & (\ack,d) &  q_2 & ~\begin{array}{l} r_1 := d \end{array} & & (\ack,8) & q_2 & 8 & \bot\\
     \hhline{-------~----}
     ~4~ & \height q_1 & q_2 & ~(\edge{\sim},r_2) = (\edge{+1},r_1)~ & (\ack,d) &  q_2 & ~\begin{array}{l} r_1 := d \end{array} & & (\ack,5) & q_2 & 5 & \bot\\
     \hhline{-------~----}
  \end{array}$
 }
 \caption{A \ngautomaton over $\SimSign$ and a run \label{table:rCRA}}
 \end{figure*}

For the language $L$ from Example~\ref{ex:rCRA}, one can show $L \not\in \ClassCMA(\SimSign)$, using an easy pumping argument. Next, we will see that \ngautomata, though more expressive than \cmautomata, are not yet enough to capture \EMSO~logic. Thus, dropping just one feature such as registers or guessing data values makes \myautomata incomparable to the logic. Assume $m = 2$ and consider $\Sign_\sim^2 = \{\edge{\sim}^1\,,\,\edge{\sim}^2\}$ (cf.\ Example~\ref{ex:sim}).

\begin{lemma}\label{lem:weakng}
$\ClassFO(\Sign_\sim^2) \,\not\subseteq\, \ClassRCRA(\Sign_\sim^2)$.
\end{lemma}

\iflong
\begin{proof}
We determine a formula $\phi \in \FO(\Sign_\sim^2)$ and show, by contradiction, that every \ngautomaton capturing \mbox{$L = L(\phi)$} will necessarily accept a data word outside $L$. Roughly speaking, $L$ consists of words where every position belongs to a pattern that is depicted in Figure~\ref{fig:pattern} and captured by the formula
$\textit{pattern}(x_1,\ldots,x_4) = x_1 \edge{\sim}^1 x_3 \mathrel{\wedge} x_1 \edge{\sim}^2 x_4 \mathrel{\wedge} x_2 \edge{\sim}^2 x_3 \mathrel{\wedge} x_2 \edge{\sim}^1 x_4$.
With this, $\phi = \forall x \exists x_1,\ldots,x_4\, (x \in \{x_1,\ldots,x_4\}
        \wedge  \mathit{pattern}(x_1,\ldots,x_4)) \in \FO(\Sign_\sim^2)$ is the formula for $L$. Suppose that there is a \ngautomaton $\Aut$ over $\Sign_\sim^2$ recognizing $L$. We build a data word $w = (a,d_1) \ldots (a,d_n) \in L$ with $n \in 4\N$  and \mbox{$\{d_1^1,\ldots,d_n^1\} \mathrel{\cap} \{d_1^2,\ldots,d_n^2\} = \emptyset$} by nesting disjoint patterns as depicted in 
Figure~\ref{fig:pattern}: we first create $i_1,\ldots,i_4$, then add $j_1,\ldots,j_4$; the next pattern is to be inserted at $n_1,\ldots,n_4$, etc. We assume that the data values of distinct patterns are disjoint. If we choose $n$ large enough, then there are an accepting run $\xi = (\state_1,\reg_1) \ldots (\state_n,\reg_n)$ of $\Aut$ on $w$ (with transition \mbox{$t_i = \trans{(\source_i,\test_i)}{a}{(\state_i,\update_i)}$} at position $i$) and positions $i_1,\ldots,i_4,j_1,\ldots,j_4$ of $w$ such that $j_1 < i_1$, 
$i_1,\ldots,i_4$ and $j_1,\ldots,j_4$ form two (disjoint) patterns, and
$t_{i_1} = t_{j_1}, \ldots, t_{i_4} = t_{j_4}$.
Now, consider the data word $w'$ that we obtain from $w$ when we swap the second data values of positions $i_1$ and $j_1$. Thus, the data part of $w'$ is \[d_1 \ldots d_{j_1 - 1} \,(d_{j_1}^1,d_{i_1}^2)\, d_{j_1+1} ~\ldots~ d_{i_1 - 1} \,(d_{i_1}^1,d_{j_1}^2)\, d_{i_1+1} \ldots d_n\,.\] We have the situation depicted in Figure~\ref{fig:forbidden}. In particular, $i_1,\ldots,i_4$ do not form a single closed cycle. This violates $\phi$, as $x=i_1$ implies $x_1 \in \{i_1,i_2\}$. Thus, $w' \not\in L$. However, applying transitions $t_1, \ldots, t_n$ still yields an accepting run $\xi' = (\state_1,\reg_1') \ldots (\state_n,\reg_n')$ of $\Aut$ on $w'$. For $i \in [n]$, $\rho_i'$ is then given as follows:
\[\rho_i'(r) = \left\{\begin{array}{ll}
    d_{i_1}^2  & \textup{~if~} \rho_i(r) = d_{j_1}^2 \textup{~and~} i \in \{j_1,j_3\}\\[0.3ex]
    d_{j_1}^2  & \textup{~if~} \rho_i(r) = d_{i_1}^2 \textup{~and~} i \in \{i_1,i_3\}\\[0.3ex]
    d_{i_1}^1  & \textup{~if~} \rho_i(r) = d_{j_1}^1 \textup{~and~} i = j_4\\[0.3ex]
    d_{j_1}^1  & \textup{~if~} \rho_i(r) = d_{i_1}^1 \textup{~and~} i = i_4\\\
    \!\!\!\rho_i(r)  & \textup{~otherwise}
\end{array}\right.\]
One can verify that $\xi'$ is indeed an accepting run on $w'$.
\qed
\end{proof}

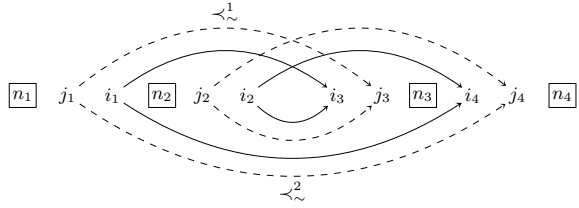
\begin{figure}[t]
\centering
{
 \scalebox{0.8}{
\begin{picture}(90,36)(-7.5,-14)
\gasset{Nh=4,Nw=3.8,AHangle=35,AHLength=1.0,AHlength=0.4,Nframe=n,Nfill=n,linewidth=0.1}


\unitlength=0.23em
\node[Nh=7,Nw=6](i1)(0,0){$j_1$}
\node[Nh=7,Nw=6](i2)(30,0){$j_2$}
\node[Nh=7,Nw=6](i3)(70,0){$j_3$}
\node[Nh=7,Nw=6](i4)(100,0){$j_4$}

\node[Nh=7,Nw=5](j1)(10,0){$i_1$}
\node[Nh=7,Nw=5](j2)(40,0){$i_2$}
\node[Nh=7,Nw=5](j3)(60,0){$i_3$}
\node[Nh=7,Nw=5](j4)(90,0){$i_4$}

\node[Nframe=y,Nmr=0,Nw=6](a1)(-10,0){$n_1$}
\node[Nframe=y,Nmr=0,Nw=6](a2)(21,0){$n_2$}
\node[Nframe=y,Nmr=0,Nw=6](a3)(79,0){$n_3$}
\node[Nframe=y,Nmr=0,Nw=6](a4)(110,0){$n_4$}

\drawedge[dash={1.5}0,curvedepth=15](i1,i3){$\edge{\sim}^1$}
\drawedge[dash={1.5}0,curvedepth=15](i2,i4){}
\drawedge[dash={1.5}0,curvedepth=-18,ELside=r](i1,i4){$\edge{\sim}^2$}
\drawedge[dash={1.5}0,curvedepth=-10](i2,i3){}

\drawedge[curvedepth=10](j1,j3){}
\drawedge[curvedepth=10](j2,j4){}
\drawedge[curvedepth=-14](j1,j4){}
\drawedge[curvedepth=-6](j2,j3){}

\end{picture}
}
}
\caption{Nested patterns\label{fig:pattern}}
\end{figure}
\begin{figure}[t]
\centering
{
 \scalebox{0.8}{
\begin{picture}(90,30)(-7.5,-12)
\gasset{Nh=4,Nw=3.8,AHangle=35,AHLength=1.0,AHlength=0.4,Nframe=n,Nfill=n,linewidth=0.1}


\unitlength=0.23em
\node[Nh=7,Nw=6](i1)(0,0){$j_1$}
\node[Nh=7,Nw=6](i2)(30,0){$j_2$}
\node[Nh=7,Nw=6](i3)(70,0){$j_3$}
\node[Nh=7,Nw=6](i4)(100,0){$j_4$}

\node[Nh=7,Nw=5](j1)(10,0){$i_1$}
\node[Nh=7,Nw=5](j2)(40,0){$i_2$}
\node[Nh=7,Nw=5](j3)(60,0){$i_3$}
\node[Nh=7,Nw=5](j4)(90,0){$i_4$}

\drawedge[dash={1.5}0,curvedepth=15](i1,i3){}
\drawedge[dash={1.5}0,curvedepth=15](i2,i4){}
\drawedge[curvedepth=-17](i1,j4){}
\drawedge[dash={1.5}0,curvedepth=-10](i2,i3){}

\drawedge[curvedepth=10](j1,j3){}
\drawedge[curvedepth=10](j2,j4){}
\drawedge[dash={1.5}0,curvedepth=-17](j1,i4){}
\drawedge[curvedepth=-6](j2,j3){}

\end{picture}
}
}
\caption{Merging patterns\label{fig:forbidden}}
\end{figure}
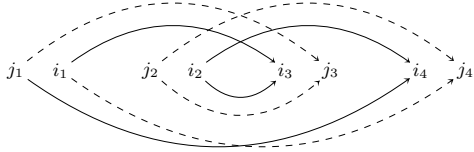
\fi

The proof of Lemma \ref{lem:weakng} can be adapted to show $\ClassFO(\CommSign) \not\subseteq \ClassRCRA(\CommSign)$. It reveals that \ngautomata can in general not detect \emph{cycles}. However, this is needed to capture $\FO$ logic \cite{Hanf1965}. In Section~\ref{sec:autvslogic}, we show that \emph{full} \myautomata capture $\FO$ and, as they are closed under projection, also $\EMSO$ logic.
Closure under projection is meant in the following sense. Let $\Gamma$ be a non-empty finite alphabet. Given $\Sign = (\RelSymb,\RelInt)$, we define another signature $\Sign_\Gamma$ for data words over $(\Sigma \times \Gamma) \times \Data^m$. Its set of relation symbols is $\{\nrel \mid \rel \in \Sign\}$. For $w \in ((\Sigma \times \Gamma) \times \Data^m)^\ast$, we set $i \mathrel{\Int{\nrel}{w}} j$ iff $i \mathrel{\Int{\rel}{{\mathit{proj}_\Sigma(w)}}} j$. Hereby, the projection $\mathit{proj}_\Sigma$ just removes the $\Gamma$ component while keeping $\Sigma$ and the data values. For $\mathcal{C} \in \{\ClassCRA,\ClassRCRA,\ClassCMA\}$, we say that $\mathcal{C}(\Sign)$ is \emph{closed under projection} if, for every $\Gamma$ and $L \subseteq ((\Sigma \times \Gamma) \times \Data^m)^\ast$, $L \in \mathcal{C}(\Sign_\Gamma)$ implies $\mathit{proj}_\Sigma(L) \in \mathcal{C}(\Sign)$.

\begin{lemma}\label{lem:projection}
For every signature $\Sign$, $\ClassCRA(\Sign)$, $\ClassRCRA(\Sign)$, and $\ClassCMA(\Sign)$ are closed under union, intersection, and projection. They are, in general, not closed under complementation.
\end{lemma}

\iflong
\begin{proof}
Closure under union and intersection follows standard automata-theoretic constructions. Closure under projection holds since projection preserves the graph structure of a data word. For non-complementability, we can rely on the corresponding result for communicating automata \cite{Bollig06}. Roughly speaking, a communicating automaton is a dynamic communicating automaton with a fixed set of at least two processes $\Proc$. It can be identified as a special case of our framework: We let $m=0$, since the number of processes is fixed. Moreover, $\Sigma = \{\,\send{c}{d}\,,\rec{c}{d} \mid c,d \in \Proc$ such that $c \neq d\,\}$ is the set of actions. Finally, we define the signature $\Sign_\Proc^0 = \{\edge{\procsymb}\,,\,\fedge{\msgsymb}\}$ as the straightforward restriction of $\CommSign$ (cf.\ Example~\ref{ex:comm}) to this bounded case. Speaking in terms of our framework, \cite{Bollig06} indeed shows that \myautomata (or, as $m=0$, \cmautomata) over $\Sign_\Proc^0$ are not closed under complementation.
\qed
\end{proof}
\fi

\section{Realizability of EMSO Specifications}\label{sec:autvslogic}

In this section, we solve the realizability problem for
$\EMSO$ specifications:

\begin{theorem}\label{thm:main} For all signatures $\Sign$, $\ClassEMSO(\Sign) \subseteq \ClassCRA(\Sign)$. An automaton can be computed in elementary time and is of elementary size.
\end{theorem}

Classical procedures that translate formulas into automata follow an inductive approach, use two-way mechanisms and tools such as pebbles, or rely on reductions to existing translations. There is no obvious way to apply any of these techniques to prove our theorem.

We therefore follow a technique from \cite{Bollig06}, which is based on ideas from \cite{ThoPOMIV96,SchwentickB99}. We first transform the first-order kernel of the formula at hand into a normal form due to Hanf \cite{Hanf1965}. According to that normal form, satisfaction of a first-order formula wrt.\ data word $w$ only depends on the spheres that occur in $\Graph{w}$, and on how often they occur,
counted up to a threshold. The size of a sphere is bounded by a radius that depends on the formula. The threshold can be computed from the radius and $|\Sign|$. We can indeed
apply Hanf's Theorem, as the structures that we consider have \emph{bounded
  degree}: every node/word position has at most $|\Sign|$ incoming and at most
$|\Sign|$ outgoing edges. In a second step, we transform the formula in normal form into a \myautomaton.

Recall that $\Sph{B}{\graph}{i}$ denotes the $B$-sphere of graph/data word $\graph$ around $i$ (cf.\ Section~\ref{sec:definitions}). Its size (number of nodes) is bounded by $\mathit{maxSize} \df (2|\Sign|+2)^{B}$. Let $\Spheres{B} = \{\Sph{B}{\graph}{i} \mid \graph=(V,\ldots)$ is an $\Sign$-graph and $i \in V\}$. We do not distinguish between isomorphic structures so that $\Spheres{B}$ is finite.

\begin{theorem}[cf.\ \cite{Hanf1965,BK2011}]\label{thm:Hanf}
Let $\phi \in \FO(\Sign)$. One can compute, in elementary time, $B \in \N$ and a boolean formula $\beta$ over $\{\,\textup{`}\sphere \le N\textup{'} \mid \sphere \in \Spheres{B}$
and $N \in \N\}$ such that $L(\phi)$ is the set of data words that satisfy $\beta$.  Here, we say that $w = w_1 \ldots w_n$ satisfies atom $\sphere \le N$ iff $|\{i \in [n] \mid \Sph{B}{w}{i} \cong \sphere\}| \le N$. The radius $B$ and the size of $\beta$ and its constants $N$ are elementary in $|\phi|$ and $|\Sign|$.
\end{theorem}

\iflong
\begin{proof}
A simple but crucial observation is that there exists a first-order sentence that is
equivalent to $\phi$ but talks about $\Graph{w}$ rather than
$w$. We simply write \mbox{$\lambda(x) = a$} instead of \mbox{$\letter{x} = a$}, and
$\bigvee_{\eta \in \mathcal{P}} \nu(x) = \eta$ instead of $\dvalue{x}{k} =
\dvalue{x}{l}$ where $\mathcal{P} \subseteq \Part{m}$ is the set of partitions of $[\param]$ such
that $k$ and $l$ occur in the same set. As $\MSO(\Sign)$-formulas cannot
distinguish between data words that induce the same graph, the boolean formula $\beta$ in normal form exists due to \cite{Hanf1965}. Actually, $\beta$ can be computed in triply exponential time \cite{BK2011}.
\qed
\end{proof}
\fi

By Theorem~\ref{thm:Hanf}, it will be useful to have a \myautomaton that, when
reading a position $i$ of data word $w$, outputs the sphere of $w$ around $i$.
Its construction is actually the main difficulty in the proof of Theorem~\ref{thm:main}, as spheres have to be computed ``in one go'', i.e., reading the word from left to right, while
accessing only certain configurations from the past.

\begin{proposition}\label{prop:sphereaut}
  Let \mbox{$B \in \N$}. One can compute, in elementary time, a \myautomaton $\Aut_B = (\States,\Reg,\Trans,{(\Local_\rel)}_{\rel \in \Sign},\mytrue)$ over $\Sign$, as well as a mapping $\pi: \States \to \Spheres{B}$ such that $L(\Aut_B) = (\Sigma \times \Data^\param)^\ast$ and, for every data word $w=w_1 \ldots w_n$, every accepting run $(q_1,\reg_1) \ldots
  (q_n,\reg_n)$ of $\Aut_B$ on $w$, and every $i \in [n]$,
  $\pi(q_i) \cong \Sph{B}{w}{i}$. Moreover, $|\States|$ and $|\Reg|$ are elementary in $B$ and $|\Sign|$.
\end{proposition}

The proposition is proved below. Let us first show how we can use it, together with Theorem~\ref{thm:Hanf}, to translate an $\EMSO$ formula into a \myautomaton.

\begin{proof}[of Theorem~\ref{thm:main}]
Let $\phi = \exists X_1 \ldots \exists X_n\, \psi \in
\EMSO(\Sign)$ be a sentence with $\psi \in \FO(\Sign)$ (we also assume $n \ge 1$). Since Theorem~\ref{thm:Hanf} applies to first-order formulas only, we extend $\Sigma$ to $\Sigma \times \Gamma$ where $\Gamma = 2^{\{1,\ldots,n\}}$. Consider the extended signature $\myhat{\Sign}$ (cf.\ Section~\ref{sec:CRA}). From $\psi$, we obtain a formula $\smash{\myhat{\psi}} \in \FO(\smash{\myhat{\Sign}})$ by replacing $\letter{x} = a$ with $\bigvee_{M \in \Gamma} \letter{x} = (a,M)$ and $x \in X_j$ with $\bigvee_{\substack{a \in \Sigma,\,M \in \Gamma}} \letter{x} = (a,M \cup \{j\})$.
Consider the radius $B \in \N$ and the normal form $\beta_\Gamma$ for $\myhat{\psi}$ due to Theorem~\ref{thm:Hanf}.
Let $\Aut_B = (\States,\Reg,\Trans,{(\Local_\rel)}_{\rel \in \myhat{\Sign}},\mytrue)$ be the \myautomaton over $\Sign_\Gamma$ from Proposition~\ref{prop:sphereaut} and $\pi$ be the
associated mapping. The global acceptance
condition of $\Aut_B$ is obtained from $\beta_\Gamma$ by replacing every atom $\sphere \le N$ with $\pi^{-1}(\sphere) \le N$ (which can be expressed as a suitable boolean formula). We hold $\Aut_B'$, a \myautomaton satisfying $L(\Aut_B') = L(\myhat{\psi})$. Exploiting closure under projection (Lemma~\ref{lem:projection}), we obtain a \myautomaton over $\Sign$ that recognizes $L(\phi) = \mathit{proj}_\Sigma(L(\myhat{\psi}))$.
\qed
\end{proof}

\paragraph{The Sphere Automaton.}

In the remainder of this section, we construct the \myautomaton $\Aut_B = (\States,\Reg,\Trans,{(\Local_\rel)}_{\rel \in \Sign},\mytrue)$ from Proposition~\ref{prop:sphereaut}, together with $\pi: \States \to \Spheres{B}$. The idea is that, at each
position $i$ in the data word $w$ at hand, $\Aut_B$ guesses the $B$-sphere
$\sphere$ of $w$ around $i$. To verify that the guess is correct, i.e., $\sphere
\cong \Sph{B}{w}{i}$, $\sphere$ is passed to each position that is connected to $i$ by an edge in $\Graph{w}$. That new position locally checks label and data
equalities imposed by $\sphere$, then also forwards $\sphere$ to its
neighbors, and so on. Thus, at any time, several local patterns have to be validated
simultaneously so that a state $q \in Q$ is actually a \emph{set} of spheres.
In fact, we consider \emph{extended} spheres $\esphere = (\sphere,\sactive,\scolor)$
where $\sphere = (U,(\rel^\esphere)_{\rel \in \Sign},\labelone,\labeltwo,\scenter)$ is
a sphere (with universe $U$ and sphere center $\gamma$), $\sactive \in U$ is
the \emph{active node}, and $\scolor$ is a color from a finite set, which will be specified later. The active node $\alpha$ indicates the current context, i.e., it corresponds to the position currently read.

Let $\eSpheres{B}$ denote the set of extended spheres, which is finite up to isomorphism. For $\esphere = (\sphere,\sactive,\scolor) \in \eSpheres{B}$, $\sphere = (U,(\rel^\esphere)_{\rel \in \Sign},\labelone,\labeltwo,\scenter)$, and $j \in U$, we let
$\esphere[j]$ refer to the extended sphere $(\sphere,j,\scolor)$ where the
active node $\alpha$ has been replaced with $j$. Now suppose that the state $q$ of
$\Aut_B$ that is reached after reading position $i$ of data word $w$ contains
$\esphere = (\sphere,\sactive,\scolor)$. Roughly speaking, this means that the
neighborhood of $i$ in $w$ shall look like the neighborhood of $\alpha$ in
$\sphere$. Thus, if $\sphere$ contains $j'$ such that $\sactive \rel^\esphere j'$,
then we must find $i'$ such that $i \rel^w i'$ in the data word. Local
final states will guarantee that $i'$ indeed exists. Moreover, the state assigned
to $i'$ in a run of $\Aut_B$ will contain the new proof obligation
$\esphere[j']$ and so forth. Similarly, an edge in (the graph of) $w$ has to be present in
spheres, unless it is beyond their scope, which is limited by $B$. All this is reflected below, in conditions T2--T6 of a transition.

We are still facing two major difficulties. Several \emph{isomorphic} spheres have to be verified simultaneously, i.e., a state must be allowed to include isomorphic spheres in different contexts. A solution to this problem is provided by the additional coloring $\scolor$. It makes sure that centers of overlapping isomorphic spheres with different colors refer to distinct nodes in the input word. To put it differently, for a given position $i$ in data word $w$, there may be $i'$ such that $0 < \mydist^w(i,i') \le 2B + 1$ and $\Sph{B}{w}{i} \cong \Sph{B}{w}{i'}$. Fortunately, there cannot be more than $(2|\Sign|+1) \cdot \mathit{maxSize}^2$ such positions. As a consequence, the coloring $\scolor$ can be restricted to the set $\{1,\ldots,(2|\Sign|+1) \cdot \mathit{maxSize}^2 +1\}$.

Implementing these ideas alone would do without registers and yield a \cmautomaton. But this cannot work due to Lemma~\ref{lem:weakng}. Indeed, a faithful simulation of cycles in spheres has to make use of data values. They need to be anticipated, stored in registers, and locally compared with current data values from the input word. We introduce a register $(\esphere,k)$ for every extended sphere $\esphere$ and $k \in [m]$. To get the idea behind  this, consider a run $(q_1,\rho_1) \ldots (q_n,\rho_n)$ of $\Aut_B$ on $w = (a_1,d_1) \ldots (a_n,d_n)$. Pick a position $i$ of $w$ and suppose that $\esphere = (U,(\rel^\esphere)_{\rel \in \Sign},\labelone,\labeltwo,\scenter,\sactive,\scolor) \in q_i$. If $\sactive$ is minimal in $E$, then there is no pending requirement to check. Now, as $\sactive$ shall correspond to the current position $i$ of $w$, we write, for every $k \in [m]$, $d_i^k$ into register $(E,k)$ (first case of T8 below).
\iflong
For all $j \in U \setminus \{\sactive\}$, on the other hand, we anticipate data values  and store them in $(E[j],k)$ (also first case of T8). They will be forwarded (second case of T8) and checked later against both the guesses made at other minimal nodes of $\esphere$ (guard $g_3$ of T7) and the actual data values in $w$ (guard $g_2$). This procedure makes sure that the values that we carry along within an accepting run agree with the actual data values of $w$.
\fi

Now, as $\cpred{\rel}^w$ and $\csucc{\rel}^w$ are monotone wrt.\ positions with identical labels and data values, two isomorphic cycles cannot be ``merged'' into one larger one, unlike in \ngautomata where different parts may act erroneously on the assumption of inconsistent data values (cf.\ Lemma~\ref{lem:weakng}). As a consequence, spheres are correctly simulated by the input word.

\smallskip

Let us formalize $\Aut_B=(\States,\Reg,\Trans,{(\Local_\rel)}_{\rel \in \Sign},\mytrue)$ and the mapping $\pi: \States \to \Spheres{B}$, following the above ideas. The set of registers is $R = \eSpheres{B} \times [\param]$. A state from $Q$ is a non-empty set $q \subseteq \eSpheres{B}$ such that
\begin{itemize}\itemsep=0.5ex
\item[(i)] there is a unique
$\esphere=(U,(\rel^\esphere)_{\rel \in \Sign},\labelone,\labeltwo,\scenter,\sactive,\scolor) \in q$ such that $\scenter = \sactive$ (we set $\pi(q) = (U,(\rel^\esphere)_{\rel \in \Sign},\labelone,\labeltwo,\scenter)$ to obtain the mapping required by Prop.~\ref{prop:sphereaut}),
\item[(ii)] there are $a \in
\Sigma$ and $\eta \in \Part{m}$ such that, for all $E=(\ldots,\lambda,\nu,\ldots) \in q$, we have $\lambda(\sactive) = a$ and $\nu(\sactive) = \eta$ (we let $\slabel(q) = a$ and $\sguard(q) = \eta$), and
\item[(iii)] for every
$(\sphere,\sactive,\scolor),(\sphere,\sactive',\scolor) \in q$, we have
$\sactive = \sactive'$.
\end{itemize}

Before we turn to the transitions, we introduce some notation. Below, $\esphere$ will always denote $(\sphere,\sactive,\scolor)$ with $\sphere = (U,(\rel^\esphere)_{\rel \in \Sign},\labelone,\labeltwo,\scenter)$; in particular, $\sactive$ refers to the active node of $\esphere$. The mappings $\csucc{\rel}^\esphere$, $\prev{\rel}^\esphere$, and $\mydist^E$ are defined for extended spheres in the obvious manner. For $j \in U$, we set $\ptype{j} = \{\rel \in \Sign \mid \prev{\rel}^\esphere(j)$ is defined$\}$. Let us fix, for all $E \in \eSpheres{B}$ such that $\ptype{\alpha} \neq \emptyset$, some arbitrary $\rel_E \in \ptype{\alpha}$. Finally, for state $q$ and $k_1,k_2 \in [m]$, we write $k_1 \sim_q k_2$ if there is $K \in \sguard(q)$ such that $\{k_1,k_2\} \subseteq K$.

\smallskip

We have a transition
$\trans{(\source,\test)}{a}{(\state,\update)}$ iff the following hold:
\begin{itemize}\itemsep=0.5ex
\item[T1] $\slabel(q) = a$

\item[T2] \parbox{9.8em}{for all $\rel \in \Sign$, $\esphere \in
  q$\,:} $\rel \not\in \dom(\source) ~\Longrightarrow~ \prev{\rel}^\esphere(\sactive) \textup{ is undefined}$
  
\item[T3] \parbox{16em}{for all $\rel \in \dom(\source)$, $\esphere \in q$, $j \in U$:}
   \parbox{11em}{$j \rel^\esphere \sactive ~\Longleftrightarrow~ \esphere[j] \in \source_\rel$}
   
\item[T4] \parbox{16em}{for all $\rel \in \dom(\source)$, $\esphere \in \source_\rel$, $j \in U$:}
   \parbox{11em}{$\sactive \rel^\esphere j ~\Longleftrightarrow~ \esphere[j] \in q$}
  
\item[T5] \parbox{11.9em}{for all $\rel \in \dom(\source)$, $\esphere \in q$\,:}
    $\parbox{9.3em}{\hfill$\prev{\rel}^\esphere(\sactive) \textup{~undefined~}$} \Longrightarrow~ \sdist^\esphere(\gamma,\sactive) = B$
\item[T6] \parbox{11.9em}{for all $\rel \in \dom(\source)$, $\esphere \in \source_\rel$:}
  $\parbox{9.3em}{\hfill$\nextwp{\rel}^\esphere(\sactive) \textup{~undefined~}$} \Longrightarrow~ \sdist^\esphere(\gamma,\sactive) = B$\vspace{1ex}
  

\iflong
\item[T7] $\test = \test_1 \wedge \test_2 \wedge \test_3$ where\\
\begin{itemize}
\item[]\hspace{-3em}
$\test_1 = \underset{\substack{k_1,k_2 \in [m]\\k_1 \,\sim_q\, k_2}}{\bigwedge}\hspace{-0.5em} k_1 = k_2
\,\mathrel{\wedge} \underset{\substack{k_1,k_2 \in [m]\\k_1 \,\not\sim_q\, k_2}}{\bigwedge} \hspace{-0.5em}\neg\,(k_1 = k_2)$ \qquad $\test_2 =
\underset{\substack{k \in [m]~\esphere \,\in\, q\\\rel \in \ptype{\sactive}}}{\bigwedge}\hspace{-0.5em} k = (\rel,(\esphere,k))$
\\
\item[]\hspace{-3em}
$\test_3 = \hspace{-1em}\underset{\substack{k \in [m]~~\esphere \,\in\, q ~~ j \,\in\, U\\\rel_1,\rel_2 \in \ptype{\sactive}}}{\bigwedge}\hspace{-1em} (\rel_1,(\esphere[j],k)) = (\rel_2,(\esphere[j],k))$
\medskip
\end{itemize}
\fi

\item[T8] for all $k \in [m]$ and $\esphere \in \eSpheres{B}$\,:
\[\update((E,k)) = \left\{\begin{array}{ll}
    (k,\sdist^\esphere(j,\sactive))   &~  \textup{if~} \exists j \in U: E[j] \in q \textup{~and~}\ptype{j} = \emptyset\\

    (\rel_{E[j]},(E,k))   &~ \textup{if~} \exists j \in U: E[j] \in q \textup{~and~}\ptype{j} \neq \emptyset\\
    \textup{undefined}   &~ \textup{otherwise}
\end{array}\right.\]
\end{itemize}

\smallskip

For every $\rel \in \Sign$, the local acceptance condition is given by $F_\rel = \{q \in Q \mid$ for all $\esphere \in q$, $\nextwp{\rel}^\esphere(\sactive)$ is undefined$\}$. Recall that the global one is $\mytrue$.

As the maximal size of a sphere is exponential in $B$ and polynomial in $|\Sign|$, the numbers $|Q|$ and $|\Reg|$ are elementary in $B$ and $|\Sign|$. Note that $\Aut_B$ can actually be constructed in elementary time.



\iflong
In the appendix, we show that the construction of $\Aut_B$ and $\pi$ is correct in the sense of Proposition~\ref{prop:sphereaut}.
\fi

\section{From Automata to Logic}
\label{sec:specialcases}

Next, we give translations from automata back to logic. Note that $\ClassEMSO(\RaSign) \subsetneqq \ClassCRA(\RaSign)$, as $\EMSO(\RaSign)$ cannot reason about data values.
However, we show that the behavior of a \myautomaton is always MSO definable and, in a sense, ``regular''. There are natural finite-state automata that do not share this property: two-way register automata (even deterministic ones) over one-dimensional data words are incomparable to $\extMSO(\SimSign)$ \cite{Neven2004}.

\begin{theorem}\label{thm:logaut}
For every signature $\Sign$, we have $\ClassCRA(\Sign) \subseteq \extClassMSO(\Sign)$.
\end{theorem}

\iflong
\begin{proof}
As usual, second-order
variables are used to encode an assignment of positions to transitions, which
is then checked for being an accepting run. To simulate register contents, we
extend a technique from \cite{Neven2004}. Let us describe how a \myautomaton $\Aut = (\States,\Reg,\Trans,{(\Local_\rel)}_{\rel \in \Sign},\Acc)$ over $\Sign$ is translated into an $\extMSO(\Sign)$-sentence $\phi_{\!\Aut}$ such that $L(\phi_{\!\Aut}) = L(\Aut)$.
Suppose $\mathcal{B}$ is the maximum of all $B$ for which there is a transition $\trans{(\source,\test)}{a}{(\state,\update)} \in \Delta$ with $f(r)=(k,B)$, for some $r$ and $k$.

\newcommand{\xupdate}{x_{\textup{u}}}
\newcommand{\xvalue}{x_{\textup{v}}}

We assume a second-order variable $X_\delta$ for every transition $\delta \in \Trans$. Moreover, we assume a variable $X_{r,B}^{\beta}$ for each $r \in \Reg$, $B \in \{1,\ldots,\mathcal{B}\}$, and each formula $\beta(\xupdate,\xvalue) \in \FO(\Sign)$, with free variables $\xupdate$ and $\xvalue$, that is of the form
\[\beta(\xupdate,\xvalue) = \exists x_1,\ldots,x_B\,(\xupdate \bowtie_1 x_1 \bowtie_2 \ldots \bowtie_B x_B = \xvalue)\]
where $\mathord{\bowtie_i} \in \{\,=\,,\,\rel\,,\,\rel^{-1} \mid \rel \in \Sign\}$. The intuition of these variables is as follows. If a position $x$ is contained in $X_\delta$ with $\delta = \trans{(\source,\test)}{a}{(\state,\update)}$ and $f(r) = (k,B)$, then $x$ will also be contained in some $X_{r,B}^{\beta}$, meaning that $x$ executes $\delta$ and the new data value of $r$ is the $k$-th data value at the unique $y$ such that $\beta(x,y)$ is satisfied.

\smallskip
The formula $\phi_{\!\Aut}$ will be of the form $\smash{\exists {(X_\delta)}_{\delta}\, \exists{(X_{r,B}^{\beta})}_{r,B,\beta}\, (\psi_1 \wedge \psi_2)}$. Here, $\psi_1 \in \FO(\Sign)$ checks whether the following hold:
\begin{itemize}
\item each position $x$ is contained in exactly one set $X_\delta$
\item for all $x$ and $r \in \Reg$, $x$ is contained in at most one set of the form $X_{r,B}^\beta$
\item if $x \in X_\delta$ with $\delta = \trans{(\source,\test)}{a}{(\state,\update)}$ and $f(r) = (k,B)$, then $x \in X_{r,B}^{\beta}$ for some $\beta$
\item the label at position $x \in X_\delta$ corresponds to the label of $\delta$
\item conditions (1) and  (2) in the definition of a run are met
\item the (potential) run is accepting, i.e., $\Local_\rel$ and $\Acc$ are respected
\end{itemize}

It remains to define $\psi_2 \in \extMSO(\Sign)$ to check property (3) of a run. This can be done by means of formulas $\psi_g(x)$, one for each atomic guard $\guard \in \{\,\theta_1 = \theta_2 \mid \theta_1,\theta_2 \in [m] \mathrel{\cup} (\Sign \mathrel{\times} \Reg)\}$. We restrict here to $\guard = ((\rel,r) = l)$ with $l \in [m]$. The other cases are similar.
Formula $\psi_\guard(x)$ checks if the contents of $r$ at position
$\cpred{\rel}(x)$ equals the $l$-th data value at $x$. It
will be of the form $\exists X\, \exists {(X_r)}_{r \in
  R}\, \chi_\guard$. The idea is that the positions in $X$ describe a path $x_1 \mathrel{\rel_1} x_2 \mathrel{\rel_2} \ldots \mathrel{\rel_{n-1}} x_n \mathrel{\rel} x$ that ``transports'' the data value $\dvalue{x}{l}$. We suppose that every position $x_i$ is contained in precisely one set $X_{r_i}$ meaning that register $r_i$ is updated by the contents of $r_{i-1}$ at position $x_{i-1}$. More precisely, we require that, for all $i \in \{2,\ldots,n\}$, there is a transition $\delta$ with register-update mapping $\update$ such that $x_i \in X_\delta$ and $\update(r_i) = (\rel_{i-1},r_{i-1})$. The last update should
concern $r$, i.e., we require $x_n \in X_r$. So suppose $x_1 \in X_{r_1}$.
It remains to ensure that register $r_1$, at $x_1$, obtains the value $\dvalue{x}{l}$. More precisely, there should be a transition $\delta$ with update mapping $\update$, as well as $k,B,\beta$ and a position $x_0$ such that $\beta(x_1,x_0)$ holds, $\update(r_1) = (k,B)$, $x_1 \in X_\delta \mathrel{\cap} X_{r_1,B}^\beta$ and $\dvalue{x_0}{k} = \dvalue{x}{l}$.

Note that $\chi_\guard$ can be defined as an $\extFO(\Sign)$-formula and $\psi_\guard(x)$
holds iff the register contents of $r$ at $\prev{\rel}(x)$ equals $\dvalue{x}{l}$.
\qed
\end{proof}
\fi

\noindent
In the proof, the non-local predicate $\dvalue{x}{k} = \dvalue{y}{l}$ is indeed essential to simulate register assignments, as we need to compare data values at positions where registers are updated. For one-dimensional data words, however, the predicate can be easily defined in $\MSO(\SimSign)$. The following theorem is dedicated to this classical setting over $\SimSign$.

\begin{theorem}\label{thm:onedata} We have the inclusions depicted 
in Figure~\ref{fig:hierarchy}.
Here, $\longrightarrow$ means `strictly included' and $\dashrightarrow$ means `included'.
\end{theorem}

\begin{figure}[t]
  \centering
{
  \scalebox{0.9}{
\begin{picture}(40,47)(-5,-2)
\gasset{Nh=6,Nw=18,Nadjustdist=1,AHangle=45,AHLength=1.3,AHlength=0.3,Nframe=n,Nfill=n,linewidth=0.2}
\gasset{linegray=0.7,Nmr=0,Nadjust=wh,Nframe=y}
\unitlength=0.22em

\node(MSO)(20,60){\parbox{14em}{\centering$\extClassMSO(\SimSign) ~=~ \ClassMSO(\SimSign)$}}

\node[Nw=29](CRA)(45,45){\parbox{6em}{$\ClassCRA(\SimSign)$}}
\node(RCRA)(45,30){\parbox{6em}{$\ClassRCRA(\SimSign)$}}

\node(extEMSO)(-3.7,45){\parbox{6.5em}{$\extClassEMSO(\SimSign)$}}
\node(EMSO)(-3.7,30){\parbox{6.5em}{$\ClassEMSO(\SimSign)$}}

\node(EMSO2)(20,15){\parbox{17em}{$\ClassEMSOtwo ~=~ \ClassCMA(\SimSign)$}}
\node(CRAra)(20,0){\parbox{12em}{\centering$\ClassCRA(\RaSign) ~=~\ClassRCRA(\RaSign)$}}

\gasset{linegray=0,Nframe=n}

\node(B)(31,45){}
\node(A)(8,33){}
\drawedge[ELside=r,ELpos=75,ELdist=1,dash={1.5}0](A,B){\scalebox{0.8}{Thm.~\ref{thm:main}}}

\node(B)(-3.7,57.5){}
\node(A)(-3.7,47.5){}
\drawedge[dash={1.5}0](A,B){}

\node(B)(45,57.5){}
\node(A)(45,47.5){}
\drawedge(A,B){}

\node(B)(45,42.5){}
\node(A)(45,32.5){}
\drawedge[dash={1.5}0](A,B){}

\node(B)(45,27.5){}
\node(A)(45,17.5){}
\drawedge(A,B){}

\node(B)(-3.7,42.5){}
\node(A)(-3.7,32.5){}
\drawedge[dash={1.5}0](A,B){}

\node(B)(-3.7,27.5){}
\node(A)(-3.7,17.5){}
\drawedge(A,B){}

\node(B)(20,12.5){}
\node(A)(20,2.5){}
\drawedge[ELside=r,ELdist=1](A,B){\scalebox{0.8}{\cite{Bjorklund10}}}

\node(eq)(65,15){\scalebox{0.8}{\cite{Bjorklund10,Bojanczy06}}}
\end{picture}
}
}
\caption{A hierarchy of automata and logics over one-dimensional data words\label{fig:hierarchy}}
\end{figure}
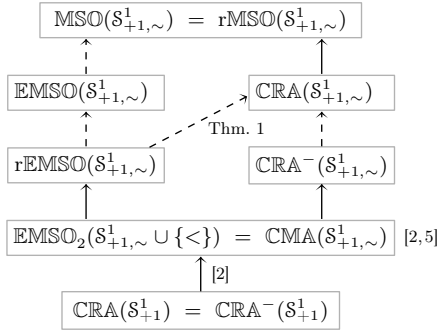

\iflong
\begin{proof}
The inclusion $\ClassEMSO(\SimSign) \subseteq \ClassCRA(\SimSign)$ is due to Theorem~\ref{thm:main}, and $\ClassCRA(\SimSign) \subseteq \extClassMSO(\SimSign)$ is
due to Theorem~\ref{thm:logaut}. The equality $\extClassMSO(\SimSign) = \ClassMSO(\SimSign)$ is obvious.

\medskip

$\ClassCMA(\SimSign) \,\subsetneqq\, \ClassEMSO(\SimSign)$\,:
Consider a \cmautomaton $\Aut$. As $\Aut$ is completely state-based and does not make use
  of any register, it is standard to define a sentence $\psi \in
  \EMSO(\SimSign)$ such that $L(\psi) = L(\Aut)$. It remains to show
  strictness of the inclusion.\footnote{Note that satisfiability of $\EMSO(\SimSign)$ is undecidable, whereas emptiness of class memory automata over $\SimSign$ is decidable \cite{Bojanczy06}. This already implies that there is no \emph{effective} translation of automata into formulas.} Suppose $\Sigma = \{\req,\ack\}$ and $\Data = \N$, and let
  $L = [\{(\req,1) \ldots (\req,n)(\ack,1) \ldots (\ack,n) \mid n \ge 1\}]_{{\SimSign}}$ (note that the proof also works if $\Sigma$ is a singleton). Towards a contradiction, suppose $L$ is recognized by \cmautomaton $\Aut$. As $\Aut$ has no access to registers, a run of $\Aut$ on
  $(\req,1) \ldots (\req,n)(\ack,1) \ldots (\ack,n)$ is actually a sequence of states $q_1
  \ldots q_{2n}$. If $n$ is large enough, there are positions $1 \le i < j \le
  n$ such that $q_i = q_j$. Now, we can simply exchange the data values at
  positions $i$ and $j$ without affecting acceptance. More precisely, $q_1
  \ldots q_{2n}$ is also an accepting run on the data word $(\req,1) \ldots
  (\req,i-1)(\req,j)(\req,i+1) \ldots (\req,j-1)(\req,i)(\req,j+1) \ldots (\req,n)(\ack,1) \ldots (\ack,n)$, which is not contained in $L$, a contradiction. On the other hand,
  $L$ is the conjunction $\phi_1 \wedge \phi_2$ of the following
  $\FO(\SimSign)$-sentences:
\begin{itemize}\itemsep=1ex
\item $\phi_1 = \exists x\, \mytrue \,\wedge\, \forall x\,\exists^{=1} y\,
\left(
      \begin{array}{ll}
             & x \edgesim y \mathrel{\wedge} \lformula{x}{\req} \mathrel{\wedge} \lformula{y}{\ack}\\
        \vee & y \edgesim x \mathrel{\wedge} \lformula{y}{\req} \mathrel{\wedge} \lformula{x}{\ack}
      \end{array}\right)$
\item $\phi_2 = \forall x,y\,
\left(
\begin{array}{cl}
& x \edgesucc y \,\mathrel{\wedge}\, \neg(\lformula{x}{\req} \mathrel{\wedge} \lformula{y}{\ack})\\
\rightarrow & \exists x',y'
      \left(
      \begin{array}{ll}
        & x \edgesim x' \edgesucc y' \mathrel{\wedge} y \edgesim y'\\
        \vee & x' \edgesucc y' \edgesim y \mathrel{\wedge} x' \edgesim x
      \end{array}\right)
\end{array}
\right)$
\end{itemize}
\smallskip
The first formula expresses that the word has positive length and each $\sim$ equivalence class has size two.
The second formula ensures the FIFO structure of a data word.

\medskip

$\ClassCMA(\SimSign) \,\subsetneqq\, \ClassRCRA(\SimSign)$\,:
Consider the language $L$ from the previous paragraph. It is not in $\ClassCMA(\SimSign)$. However, Example~\ref{ex:rCRA} demonstrates that there is a \ngautomaton recognizing $L$.

\medskip

$\ClassMSO(\SimSign) \,\not\subseteq\, \ClassCRA(\SimSign)$\,:
We encode grids into data words. An $(i,j)$-\emph{grid} is a graph that has a
height $i \in \N$ and a width $j \in \N$ meaning that it has $i$ rows and $j$ columns that are connected by a horizontal and a vertical immediate successor relation. Nodes are labeled by elements from $\Sigma = \{a,b,c\}$. We encode an
$(i,j)$-grid as the data word \[(a_{11},1) \ldots (a_{i1},i) (a_{12},1) \ldots
(a_{i2},i) \ldots \ldots (a_{1j},1) \ldots (a_{ij},i)\] where $a_{kl} \in \Sigma$ is the labeling of the grid  node $(k,l)$. Hereby, each subword $(a_{1k},1) \ldots
(a_{ik},i)$ constitutes a column. Then, moving down in the grid corresponds to a
$\edgesucc$-step in the data word, moving right corresponds to a
$\edgesim$-step. These steps are $\FO(\SimSign)$-definable.

Consider the set $\mathcal{L}$ of grids of the form $H_1.C.H_2$ where
$C$ is a single column of $c$-labeled nodes, and $H_1$ and $H_2$ are grids with labels from $\{a,b\}$ such that the sets of different column words (over $\{a,b\}$) in $H_1$ and $H_2$ coincide. We know that $\mathcal{L}$ is MSO-definable in the signature of a grid. Therefore, the encoding $L$ of $\mathcal{L}$ into data words is $\MSO(\SimSign)$-definable. Using an argument from \cite{ThoPOMIV96}, we show that $L \not\in \ClassCRA(\SimSign)$. First observe that the number of distinct sets of columns words over $\{a,b\}$ of length $n$ is $2^{2^n}$. Suppose, towards a contradiction, that there is 
a \myautomaton $\Aut = (\States,\Reg,\Trans,{(\Local_\rel)}_{\rel \in \Sign},\Acc)$ such that $L(\Aut) = L$. Without loss of generality, we assume that  $\Acc$ is given in terms of a simple set of global final states. In a run of $\Aut$ on the data-word encoding of grid $H_1.C.H_2$ of height $n$, all the information that $\Aut$ has about $H_1$ must be encoded in the $n$ configurations that are taken while reading the $c$-labeled positions.
The number of tuples of $n$ configurations that $\Aut$ can distinguish is bounded by
\[N = |Q|^n \cdot 2^{(|R| \cdot n)^2} \cdot (n+1)^{|R| \cdot n}\,.\]
Here, the second factor is an upper bound on the number of equivalence classes on the set $\{1,\ldots,|R| \cdot n\}$, which captures guessed values, and the third factor is the number of registers assignments. Now, as $Q$ and $R$ are fixed, $N$ does not grow sufficiently fast so that $\Aut$ will accept a data word outside $L$.

\medskip

$\ClassCRA(\RaSign) \subseteq \ClassRCRA(\RaSign)$\,:
Note first that \myautomata over $\RaSign$ are a variant of the \emph{register automata with non-deterministic reassignment} from \cite{KaminskiZ10}. The crucial difference is that 
the ``look-ahead'' of $\ClassCRA(\RaSign)$ is bounded, while the automata from \cite{KaminskiZ10} can guess any arbitrary data value. As a consequence, the latter capture the set of data words such that all data values (except the last one) are different from the last data value. We will show that, on the other hand, \myautomata over $\RaSign$ are no more expressive than classical register automata, which cannot recognize that language.

Let $\Aut = (\States,\Reg,\Trans,{(\Local_\rel)}_{\rel \in \Sign},\Acc)$ be a \myautomaton over $\RaSign$. We sketch the construction of a \ngautomaton $\Aut' = (\States',\Reg',\Trans',{(\Local_\rel')}_{\rel \in \Sign},\Acc')$ over $\RaSign$ such that $L(\Aut) = L(\Aut')$.  Let $\mathcal{B}$ be the maximal value $B$ such that an update of $\Aut$ is of the form $\update(r) = (k,B)$. Without loss of generality, we assume that $B \ge 1$ exists. The idea is that $\Aut'$ keeps track of the register contents of the last $B$ positions, and of the last $B$ data values read. To this aim, we set $\Reg' = \{-\mathcal{B},\ldots,-1\} \times (\Reg \mathrel{\uplus} \{\mathsf{current}\})$. Register $(-i,\mathsf{current})$ contains the $i$-th last input data value (wrt.\ the next position to read), and register $(-i,r)$ simulates register assignments of $\Aut$ for $r$. In particular, this allows us to access every input data value from the last $B \le \mathcal{B}$ positions. In order to anticipate data values, a state of $\Aut'$ contains, apart from a state of $\Aut$, an equivalence relation over both the new set of registers $\Reg'$ and the next $B$ positions. Thus, a state of $\Aut'$ is a pair $(q,\sim)$ where $q \in Q$ and $\sim$ is an equivalence relation over $\Reg' \times \{1,\ldots,\mathcal{B}\}$.

To simulate an update $\update(r) = (k,B)$ of $\Aut$ with $B \ge 1$, $\Aut'$
either writes the current value or one of the values stored in $(-B,\mathsf{current}),\ldots,(-1,\mathsf{current})$ into $r$, or goes into a state in which $r$ and at least one of the next $B$ positions are considered equivalent. Of course, the equivalence has to be \emph{globally consistent} and \emph{locally consistent} meaning that two equivalent registers should contain the same data value. Moreover, when $\Aut'$ is in a state where the next position and a defined register $r$ are considered equivalent, then the next symbol to read is the contents of $r$. If, in contrast, the next position is not equivalent to some defined register, then $\Aut'$ should read a data value that is currently not stored, and store it in $r$ (unless another update for $r$ applies). This finally ensures that a suitable data value in terms of an equivalence relation has been guessed when performing an update of the form $\update(r) = (k,B)$. 
\qed
\end{proof}
\fi


\noindent
The remaining (strict) inclusions are left open. When there are no data values, we have expressive equivalence of EMSO logic and \myautomata (which then reduce to \cmautomata). The translation from automata to logic follows the standard approach. The following theorem is a proper generalization of the main result of \cite{Bollig06}.

\begin{theorem}\label{thm:nodata}
Suppose $m=0$. For every signature $\Sign$, $\extClassEMSO(\Sign) = \ClassCRA(\Sign)$.
\end{theorem}

\iflong
\section{Infinite Data Words}\label{sec:infinite}

In the realm of reactive systems, it is appropriate to consider infinite data words, i.e., sequences from the set $(\Sigma \times \Data^m)^\omega$. Note that all the notions that we introduced in Section~\ref{sec:definitions} carry over to the new domain. In particular, a formula from $\MSO(\Sign)$ is interpreted over an infinite word $w$ without modifying the definition.
However, its fragment $\EMSO(\Sign)$ now appears limited. In terms of $\CommSign$, one cannot express ``some process sends infinitely many messages during an execution'', as can be shown using Hanf's Theorem. We therefore introduce a first-order quantifier $\exists^\infty$. Formula $\exists^\infty x\, \phi$ is satisfied by $w = w_1 w_2 \ldots
\in (\Sigma \times \Data^m)^\omega$ if there are infinitely many positions $i \ge 1$ such that $\phi$ is satisfied when $x$ is interpreted as $i$. We obtain the logics $\FO^\infty(\Sign)$ and $\EMSO^\infty(\Sign)$ as well as the language class $\ClassEMSO^\infty(\Sign)$. Now, a translation from logic into automata requires an extension of \myautomata. We define an $\omega$-\emphmyautomaton (over $\Sign$) to be a tuple
  $\Aut = (\States,\Reg,\Trans,{(\Local_\rel)}_{\rel \in \Sign},\Acc)$ where $\States,\Reg,\Trans,{(\Local_\rel)}_{\rel \in \Sign}$ are as in \myautomata, and 
$\Acc$ is henceforth a boolean formula over $\{\,\textup{`}q
= \infty\textup{'} \mid q \in Q\} \mathrel{\cup} \{\,\textup{`}q
\le N\textup{'} \mid q \in Q$ and $N \in \N\}$. Infnite runs $(q_1,\rho_1)(q_2,\rho_2) \ldots$ and satisfaction of the new global acceptance condition are defined as one would expect. In particular, atom $q = \infty$ is satisfied if $|\{i \ge 1 \mid q_i = q\}| = \infty$. The class of languages recognized by $\omega$-\myautomata is denoted by $\omega$-$\ClassCRA(\Sign)$. Theorems~\ref{thm:main} and \ref{thm:nodata} extend to infinite words.

\begin{theorem}\label{thm:infmain} For all $\Sign$, we have $\ClassEMSO^\infty(\Sign) \subseteq \omega\textup{-}\ClassCRA(\Sign)$. The size of the automaton is elementary in the size of the formula and $|\Sign|$. If $m = 0$, then $\ClassEMSO^\infty(\Sign) = \omega\textup{-}\ClassCRA(\Sign)$.
\end{theorem}

\begin{proof}
The crucial observation is that Proposition~\ref{prop:sphereaut} still holds. We actually take the same automaton $\Aut_B$ and run it on infinite words. The argument that makes the construction work relies on the fact that the \emph{past} of any word position is finite.
Moreover, it was shown in \cite{BK-IC08} that Theorem~\ref{thm:Hanf} has a counterpart for formulas with infinity quantifier. The proof is based on Vinner's extension of Ehrenfeucht-Fra{\"i}ss{\'e} games \cite{Vinner1972}.
Thus, for $\phi \in \FO^\infty(\Sign)$, there are $B \in \N$ and a boolean formula $\beta$ over
$\{\textup{`}\sphere = \infty\textup{'}\,,\,\textup{`}\sphere \le N\textup{'} \mid \sphere \in \Spheres{B}$
and $N \in \N\}$ such that $L(\phi)$ is the set of data words
that satisfy $\beta$. With this, the constructions from Section~\ref{sec:autvslogic} can be  adapted to translate an $\EMSO^\infty(\Sign)$-sentence into an $\omega$-\myautomaton over $\Sign$.
\qed
\end{proof}

\smallskip

We remark that the proof of Theorem~\ref{thm:infmain} is not effective. Unlike the proof of Theorem~\ref{thm:main}, it does not rely on \cite{Hanf1965,BK2011}. We do not know if there is an effective alternative.
\fi

\section{Conclusion}\label{sec:conclusion}

We studied the realizability problem for data-word languages. A particular case of this general framework constitutes a first step towards a logically motivated automata theory for dynamic message-passing systems.
\iflong
In light of this, it would be desirable to synthesize smaller and deadlock-free automata from logical or algebraic specifications. A good starting point for those studies may be temporal logic \cite{DL-tocl08,Schwentick2010}.

Our approach to modeling systems over infinite alphabets may also lead to meaningful model-checking questions.
\fi
It would be interesting to extend \cite{LeuckerMM02}, whose logic corresponds to ours in the case of $\CommSign$, to general data words.

\bibliographystyle{plain}
\bibliography{lit}

\newpage


\iflong
\appendix


\section*{A.~~Correctness of sphere automaton}

We will show that the \myautomaton $\Aut_B = (\States,\Reg,\Trans,{(\Local_\rel)}_{\rel \in \Sign},\Acc)$ over $\Sign$ and the mapping $\pi: \States \to \Spheres{B}$ are correct in the sense of Proposition~\ref{prop:sphereaut}: $L(\Aut_B) = (\Sigma \times \Data^\param)^\ast$ and, for every data
  word $w=w_1 \ldots w_n$ (where $w_i = (a_i,d_i)$), every accepting run $(q_1,\reg_1) \ldots
  (q_n,\reg_n)$ of $\Aut_B$ on $w$, and every position $i \in [n]$,
 $\pi(q_i) \cong \Sph{B}{w}{i}$.

\newcommand{\myword}{w}
\newcommand{\mylambda}{{\smash{\hat{\lambda}}}}
\newcommand{\mynu}{{\smash{\hat{\nu}}}}
\newcommand{\coloring}{\Phi}

\paragraph{Every data word is accepted.}

Let us first show $L(\Aut_B) = (\Sigma \times \Data^\param)^\ast$, i.e., that every data word is accepted by $\Aut_B$. Let $\myword = (a_1,d_1) \ldots (a_n,d_n) \in (\Sigma \times \Data^\param)^\ast$ be any data word and let $\Graph{\myword} =
([n],(\rel^\myword)_{\rel \in \Sign},\mylambda,\mynu)$ be its associated graph. We have to show $\myword \in L(\Aut_B)$. A key issue is the assignment of colors to word positions in $\myword$ such that overlapping spheres can be verified simultaneously. Let $i,i'  \in [n]$. We say that $i$ and $i'$ have a $B$\emph{-overlap} in $\myword$ if both 
$\Sph{B}{\myword}{i} \cong \Sph{B}{\myword}{i'}$ and $\mydist^\myword(i,i') \le 2B+1$.

\begin{lemma}\label{cl:coloring}
There is a mapping $\coloring: [n] \to \{1,\ldots,(2|\Sign|+1) \cdot \mathit{maxSize}^2+1\}$ such that $\coloring(i) \neq \coloring(i')$ whenever $i$ and $i'$ are distinct and have a $B$-overlap.
\end{lemma}

\begin{proof}
We obtain $\coloring$ as a coloring of the undirected graph $([n],\mathit{Arcs})$ where two nodes are connected iff they are distinct and have a $B$-overlap. The graph has degree at most 
$(2|\Sign|+1) \cdot \mathit{maxSize}^2$ so that it can be $((2|\Sign|+1) \cdot \mathit{maxSize}^2 +1)$-colored by some mapping $\coloring$, i.e., $\coloring(i) \neq \coloring(i')$ for every edge $\{i,i'\}$.
\qed
\end{proof}

\smallskip

\newcommand{\ic}{i_\textup{c}}

We now define a sequence $\xi = (q_1,\rho_1) \ldots (q_n,\rho_n)$ of configurations of $\Aut_B$ and show that $\xi$ is an accepting run of $\Aut_B$ on $\myword$. Let $i \in [n]$. We set
\[
q_i = \{\, (\Sph{B}{\myword}{\ic},i,\Phi(\ic)) ~\mid~ \ic  \in [n] \text{ such that }\mydist^\myword(\ic,i) \le B\,\}\,.
\]
Suppose $\esphere = (\sphere,\sactive,\scolor)$, $\sphere = (U,(\rel^\esphere)_{\rel \in \Sign},\labelone,\labeltwo,\scenter)$, and $k \in [m]$. We define $\rho_i((E,k))$ as follows. If there are positions $\ic,i' \in [n]$ such that $\mydist^\myword(\ic,i) \le B$, $\mydist^\myword(\ic,i') \le B$, $(\sphere,\sactive) \cong (\Sph{B}{\myword}{\ic},i')$, and $\scolor = \coloring(\ic)$, then we set $\rho_i((E,k)) = d^k(i')$. Otherwise, we let $\rho_i((E,k))$ be undefined.
Note that $\rho_i((E,k))$ is well defined, as there is at most one pair $\ic,i'$ satisfying the above properties.

We check that $q_i$ is a state. Let $\esphere = (\sphere,\sactive,\scolor) \in q_i$ and $\esphere' = (\esphere',\sactive',\scolor') \in q_i$ with $\sphere = (U,(\rel^\esphere)_{\rel \in \Sign},\labelone,\labeltwo,\scenter)$ and $\sphere' = (U',(\rel^{\esphere'})_{\rel \in \Sign},\labelone',\labeltwo',\scenter')$.
\begin{itemize}\itemsep=0.5ex
\item[(i)] Assume $\scenter = \sactive$ and $\scenter' = \sactive'$. Then, $(\sphere,\scenter) \cong (\Sph{B}{\myword}{i},i)$
and $(\sphere',\scenter') \cong (\Sph{B}{\myword}{i},i)$. Thus, $(\sphere,\scenter) \cong (\sphere',\scenter')$. Moreover, $\scolor = \scolor' = \coloring(i)$.

\item[(ii)] Clearly, we have $\labelone(\sactive) = \labelone'(\sactive')$ and $\labeltwo(\sactive) = \labeltwo'(\sactive')$.

\item[(iii)] Suppose $\sphere \cong \sphere'$ ($\sphere = \sphere'$, for simplicity) and $\scolor = \scolor'$. According to the definition of $q_i$, there are positions $i_1,i_2$ of $\myword$ such that $\mydist^\myword(i,i_1) \le B$, $\mydist^\myword(i,i_2) \le B$, $(\sphere,\sactive) \cong (\Sph{B}{\myword}{i_1},i)$, $(\sphere,\sactive') \cong (\Sph{B}{\myword}{i_2},i)$, and $\scolor = \coloring(i_1) = \coloring(i_2)$. We have $(\Sph{B}{\myword}{i_1},i) \cong (\Sph{B}{\myword}{i_2},i)$. As $i_1$ and $i_2$ have a $B$-overlap, we also have, by Lemma~\ref{cl:coloring}, $i_1 = i_2$. We deduce $\sactive = \sactive'$.
\end{itemize}

Next, we define a tuple $t_i = \trans{(\source_i,\test_i)}{a_i}{(\state_i,\update_i)}$ for all $i \in [n]$. We let $(\source_i)_{\rel} = q_{\prev{\rel}^\myword(i)}$ (which might be undefined). Moreover, let $g_i$ and $f_i$ be uniquely given by conditions T7 and T8 where we replace $q$ with $q_i$. Before we check that conditions (1)--(4) of a run are satisfied, we verify that $t_i$ is indeed a transition. In the following, we let $\esphere$ always refer to $\esphere = (\sphere,\sactive,\scolor)$ with $\sphere = (U,(\rel^\esphere)_{\rel \in \Sign},\labelone,\labeltwo,\scenter)$.

\begin{itemize}\itemsep=0.6ex
\item[T1] Obviously, we have $\slabel(q_i) = a_i$.

\item[T2] Let $\rel \in \Sign \setminus \dom(\source_i)$ (which implies that $\prev{\rel}^\myword(i)$ is undefined) and $\esphere \in q_i$. We have $(\sphere,\sactive) \cong (\Sph{B}{\myword}{\ic},i)$ for some $\ic$ with $\mydist^\myword(\ic,i) \le B$. As $\prev{\rel}^\myword(i)$ is undefined, we conclude that $\prev{\rel}^\esphere(\sactive)$ is undefined, too.

\item[T3] Let $\rel \in \dom(\source_i)$, $\esphere \in q_i$, $j \in U$, and $i_\rel = {\prev{\rel}^\myword(i)}$.\smallskip

Suppose $j \rel^\esphere \sactive$. We need to show $\esphere[j] \in q_{i_\rel}$. As $E \in q_i$, there is $\ic \in [n]$ such that $\mydist^\myword(\ic,i) \le B$, $(\sphere,\sactive) \cong (\Sph{B}{\myword}{\ic},i)$, and $\scolor = \coloring(\ic)$. Since $\sdist^\esphere(\scenter,j) \le B$ implies $\mydist^\myword(\ic,i_\rel) \le B$, and since $(\sphere,j) \cong (\Sph{B}{\myword}{\ic},i_\rel)$ and $\scolor = \coloring(\ic)$, we deduce $\esphere[j] = (\sphere,j,\scolor) \in q_{i_\rel}$.\smallskip

Conversely, suppose $\esphere[j] \in q_{i_\rel}$. We shall show $j \rel^\esphere \sactive$. There are positions $\ic,\ic' \in [n]$ such that we have $\mydist^\myword(\ic,i) \le B$, $\mydist^\myword(\ic',i_\rel) \le B$, $(\sphere,\sactive) \cong (\Sph{B}{\myword}{\ic},i)$, $(\sphere,j) \cong (\Sph{B}{\myword}{\ic'},i_\rel)$, and $\scolor = \coloring(\ic) = \coloring(\ic')$. Note that $\ic$ and $\ic'$ have a $B$-overlap. By Lemma~\ref{cl:coloring}, $\ic = \ic'$. As, then, $(\sphere,j) \cong (\Sph{B}{\myword}{\ic'},i_\rel)$, $(\sphere,\sactive) \cong (\Sph{B}{\myword}{\ic'},i)$, and $i_\rel \rel^\myword i$, we can deduce $j \rel^\esphere \sactive$.

\item[T4] is shown similarly to T3.

\item[T5] Let $\rel \in \dom(\source_i)$ and $\esphere \in q_i$ such that $\prev{\rel}^\esphere(\sactive)$ is undefined. There is $\ic \in [n]$ such that $\mydist^\myword(\ic,i) \le B$ and $(\sphere,\sactive) \cong (\Sph{B}{\myword}{\ic},i)$. Now, suppose $\sdist^\esphere(\gamma,\sactive) < B$. But then, we also have $\mydist^\myword(\ic,i) < B$ and $\prev{\rel}^\esphere(\sactive)$ is defined,  a contradiction. We deduce that $\sdist^\esphere(\gamma,\sactive) = B$.

\item[T6] is shown similarly to T5.

\item[T7] and T8 are immediate.
\end{itemize}
So far, we know that $t_i$ is a transition. Now, let us check the run conditions.
\begin{itemize}\itemsep=0.6ex
\item[(1)] and (2) are readily verified.
\item[(3)]Consider guard $\test_i = \test_1 \mathrel{\wedge} \test_2 \mathrel{\wedge} \test_3$.
We first check subformula $\test_1$. For $k_1,k_2 \in [m]$, by the definition of $\sim_{q_i}$ and $\Graph{w}$, $k_1 \sim_{q_i} k_2$ iff $d_i^{k_1} = d_i^{k_2}$.
Now, consider $\test_2$ and an atomic subformula $k = (\rel,(E,k))$ where $k \in [m]$, $\esphere \in q$, and $\rel \in \ptype{\sactive}$. Set $i_\rel = \prev{\rel}^\myword(i)$, which must indeed exist (by T2). As $E \in q_i$, there is $\ic \in [n]$ such that $\mydist^\myword(\ic,i) \le B$, $(\sphere,\sactive) \cong (\Sph{B}{\myword}{\ic},i)$, and $\scolor = \coloring(\ic)$. This implies $\mydist^\myword(\ic,i_\rel) \le B$, and we obtain $\rho_{i_\rel}((E,k)) = d_{i}^k$ so that $\test_2$ also holds.
Finally, we have to check $\test_3$. Consider its subformula $(\rel_1,(\esphere[j],k)) = (\rel_2,(\esphere[j],k))$ where $k \in [m]$, $E \in q_i$, $j \in U$, and $\rel_1,\rel_2 \in \ptype{\alpha}$. Let $i_1 = \prev{\rel_1}^\myword(i)$ and $i_2 = \prev{\rel_2}^\myword(i)$ (they both exist). Moreover, let $j_1 = \prev{\rel_1}^\esphere(\sactive)$ and $j_2 = \prev{\rel_2}^\esphere(\sactive)$.
As $E \in q_i$, there is $\ic \in [n]$ such that $\mydist^\myword(\ic,i) \le B$, $(\sphere,\sactive) \cong (\Sph{B}{\myword}{\ic},i)$, and $\scolor = \coloring(\ic)$. Due to the isomorphism, there is a unique $i' \in [n]$ such that $\mydist^\myword(\ic,i') \le B$ and $(\sphere,j) \cong (\Sph{B}{\myword}{\ic},i')$. Moreover, we have 
$(\sphere,j_1) \cong (\Sph{B}{\myword}{\ic},i_1)$ and $(\sphere,j_2) \cong (\Sph{B}{\myword}{\ic},i_2)$. In particular, $\mydist^\myword(\ic,i_1) \le B$ and $\mydist^\myword(\ic,i_2) \le B$. We deduce $\rho_{i_1}((E[j],k)) = \rho_{i_2}((E[j],k)) = d_{i'}^k$. Thus, $\test_3$ is satisfied.

\item[(4)] Let $(E,k) \in \Reg$. We distinguish three cases. 

\begin{itemize}

\item If there is $j \in U$ such that $E[j] \in q_i$, and $\ptype{j} \neq \emptyset$, then we have $\update_i((E,k)) = (\rel,(E,k))$ with $\rel = \rel_{E[j]}$. Since $E[j] \in q_i$, there is a position $\ic \in [n]$ such that $\mydist^\myword(\ic,i) \le B$, $(\sphere,j) \cong (\Sph{B}{\myword}{\ic},i)$, and $\scolor = \coloring(\ic)$. Moreover, there is a unique position $i'  \in [n]$ such that $\mydist^\myword(\ic,i') \le B$ and $(\sphere,\sactive) \cong (\Sph{B}{\myword}{\ic},i')$. As \mbox{$j_{\rel} = \prev{\rel}^\esphere(j)$} is defined, \mbox{$i_{\rel} = \prev{\rel}^\myword(i)$} is defined, too. Note that $(\sphere,j_\rel) \cong (\Sph{B}{\myword}{\ic},i_\rel)$ and $\mydist^\myword(\ic,i_\rel) \le B$. We obtain $\rho_i((E,k)) = d_{i'}^k = \rho_{i_\rel}((E,k))$. \smallskip

\item If there is $j \in U$ such that $E[j] \in q_i$ and $\ptype{j} = \emptyset$, then $\update_i((E,k)) = (k,\mydist^{\esphere}(\sactive,j))$. We show $\rho_i((E,k)) \in \Data_{B'}^k(i)$ where $B' = \mydist^{\esphere}(\sactive,j)$. As $E[j] \in q_i$, there is $\ic \in [n]$ such that $\mydist^\myword(\ic,i) \le B$, $(\sphere,j) \cong (\Sph{B}{\myword}{\ic},i)$, and $\scolor = \coloring(\ic)$. Thus, there is a unique position $i' \in [n]$ such that $\mydist^\myword(\ic,i') \le B$ and $(\sphere,\sactive) \cong (\Sph{B}{\myword}{\ic},i')$. We have $\mydist^\myword(i',i) \le \mydist^{\esphere}(\sactive,j)$, and we can deduce $\rho_i((E,k)) = d_{i'}^{k} \in \Data_{B'}^k(i)$. \smallskip

\item If there is no $j \in U$ such that $E[j] \in q_i$, then $\update_i((E,k))$ is undefined. Therefore, $\rho_i((E,k))$ should be undefined, too. Suppose, towards a contradiction, that $\rho_i((E,k)) \in \Data$. Then, there are $\ic,i' \in [n]$ such that we have $\mydist^\myword(\ic,i) \le B$, $\mydist^\myword(\ic,i') \le B$, $(\sphere,\sactive) \cong (\Sph{B}{\myword}{\ic},i')$, and $\scolor = \coloring(\ic)$, But then, there is a unique $j \in U$ such that $(\sphere,j) \cong (\Sph{B}{\myword}{\ic},i)$ so that $E[j] \in q_i$, which is a contradiction.
\end{itemize}
\end{itemize}

\smallskip

We conclude that $\xi$ is a run. Let us quickly verify that it is accepting. Trivially, $\Acc = \mytrue$ is satisfied. Now suppose $\rel \in \Sign$ and consider any position $i \in [n]$ such that $\nextwp{\rel}^\myword(i)$ is undefined. We have to show that $q_i$ is contained in $F_\rel$, i.e., $\nextwp{\rel}^\esphere(\sactive)$ is undefined for all $\esphere \in q_i$. So suppose $\esphere \in q_i$. There is $\ic \in [n]$ such that 
$\mydist^\myword(\ic,i) \le B$ and $(\sphere,\sactive) \cong (\Sph{B}{\myword}{\ic},i)$. As $\nextwp{\rel}^\myword(i)$ is undefined, $\nextwp{\rel}^\esphere(\sactive)$ must be undefined, too.

\paragraph{Every run keeps track of spheres.}

In this part of the proof, we show that we can infer, from every accepting run of $\Aut_B$ on data word $w$, the spheres that occur in $\Graph{w}$.

Let $\myword = (a_1,d_1) \ldots (a_n,d_n) \in (\Sigma \times \Data^\param)^\ast$ be a data word and $\Graph{\myword} =
([n],(\rel^\myword)_{\rel \in \Sign},\mylambda,\mynu)$ its graph. Suppose $\xi = (q_1,\rho_1) \ldots (q_n,\rho_n)$ is an accepting run of $\Aut_B$ on $w$ with corresponding transitions $t_1,\ldots,t_n$ where $t_i = \trans{(\source_i,\test_i)}{a_i}{(\state_i,\update_i)}$.

\medskip

The following claim states that an arbitrarily long path of an extended sphere $\esphere$ that starts in its active node is faithfully simulated by $\myword$. It will turn out to be crucial that, hereby, the data values in registers of the form $(E[j],k)$ are invariant during that simulation.

\begin{lemma}\label{cl:simulate}
Let $i \in [n]$ be some position, $e \ge 0$, and $\esphere = (\sphere,\sactive,\scolor)  \in q_i$ with $\sphere = (U,(\rel^\esphere)_{\rel \in \Sign},\labelone,\labeltwo,\scenter)$. Suppose there are $j_0,\ldots,j_e \in U$ and $\rel_1\ldots,\rel_\entf \in \Sign$ such that $\sactive = j_0$ and, for all $\range \in \{0,\ldots,\entf-1\}$, $j_\range \mathrel{\rel_{\range+1}^{\esphere}} j_{\range+1}$ or $j_{\range+1} \mathrel{\rel_{\range+1}^{\esphere}} j_\range$. Then, there is a unique sequence $i = i_0,\ldots,i_\entf \in [n]$ such that the following hold:
\begin{itemize}\itemsep=0.5ex

\item for each $\range \in \{0,\ldots,\entf-1\}$, $j_\range \mathrel{\rel_{\range+1}^{\esphere}} j_{\range+1}$ implies $i_\range \mathrel{\rel_{\range+1}^{\myword}} i_{\range+1}$ and $j_{\range+1} \mathrel{\rel_{\range+1}^{\esphere}} j_\range$ implies $i_{\range+1} \mathrel{\rel_{\range+1}^{\myword}} i_\range$

\item for each $\range \in \{0,\ldots,\entf\}$, we have $\esphere[j_\range] \in q_{i_\range}$, $\lambda(j_\range) = a_{i_\range}$, and $\nu(j_\range) = \mynu(i_\range)$

\item for each $\range \in \{1,\ldots,\entf\}$, $k \in [m]$, and $j \in U$, we have $\rho_{i_0}((\esphere[j],k)) = \rho_{i_{\range}}((\esphere[j],k))$

\item for each $\range \in \{0,\ldots,\entf\}$ and $k \in [m]$, we have that $\rho_{i_\range}((\esphere[j_\range],k)) = d_{i_\range}^k$
\end{itemize} 
\end{lemma}

\begin{proof}
We proceed by induction on $e$. Suppose $\entf = 0$.
By T1 and guard $\test_1$ of T7, $\lambda(\sactive) = a_{i}$ and $\nu(\sactive) = \mynu(i)$.
Let $k \in [m]$ and suppose $\ptype{\sactive} \neq \emptyset$. Then, $\update_{i}((\esphere,k)) = (\rel,(\esphere,k))$ where we let  $\rel = \rel_{\esphere}$. Thus, $\rho_{i}((\esphere,k)) = \rho_{\prev{\rel}^\myword(i)}((\esphere,k))$. By guard $\test_2$ of T7, we have  $\rho_{i}((\esphere,k)) = d_{i}^k$. If $\ptype{\sactive} = \emptyset$, then $\rho_{i}((\esphere,k)) = d_{i}^k$ is due to the update $\update_{i}((\esphere,k)) = (k,0)$ (T8).

So let $\entf \ge 0$, $j_0,\ldots,j_\entf,j_{\entf+1}
  \in U$, and $\rel_1,\ldots,\rel_\entf,\rel_{\entf+1} \in \Sign$ such that $\sactive = j_0$ and, for every $\range \in \{0,\ldots,\entf\}$, $j_\range \mathrel{\rel_{\range+1}^{\esphere}} j_{\range+1}$ or $j_{\range+1} \mathrel{\rel_{\range+1}^{\esphere}} j_\range$. Let $i_0,\ldots,i_\entf
  \in [n]$ be the unique corresponding sequence with the required properties.
We consider two cases:
\begin{itemize}\itemsep=0.5ex
\item Assume $j_\entf \rel_{\entf+1}^{\esphere} j_{\entf+1}$. Then, $q_{i_\entf} \not\in F_{\rel_{\entf+1}}$ so that $\nextwp{\rel_{\entf+1}}^{\myword}(i_{\entf})$ is defined. We set $i_{\entf+1} = \nextwp{\rel_{\entf+1}}^{\myword}(i_{\entf})$.

Due to T4, we have $\esphere[j_{\entf+1}] \in q_{i_{\entf+1}}$. By T1 and guard $\test_1$ of T7, we obtain $\lambda(j_{\entf+1}) = a_{i_{\entf+1}}$, and $\nu(j_{\entf+1}) = \mynu(i_{\entf+1})$.

Let $k \in [m]$ and $j \in U$. Due to condition T8, $\esphere[j_{\entf+1}] \in q_{i_{\entf+1}}$ implies that $\update_{i_{\entf+1}}((\esphere[j],k)) = (\rel,(\esphere[j],k))$ for some $\rel \in \Sign$. Due to guard $\test_3$ of condition T7, we have $\rho_{\prev{\rel}^{\myword}(i_{\entf+1})}((\esphere[j],k)) = \rho_{i_{\entf}}((\esphere[j],k))$. We can now deduce $\rho_{i_{\entf}}((\esphere[j],k)) = \rho_{i_{\entf+1}}((\esphere[j],k))$.

Finally, let $k \in [m]$. We have $\update_{i_{\entf+1}}((\esphere[j_{\entf+1}],k)) = (\rel,(\esphere[j_{\entf+1}],k))$ where we let  $\rel = \rel_{\esphere[j_{\entf+1}]}$. Thus, $\rho_{i_{\entf+1}}((\esphere[j_{\entf+1}],k)) = \rho_{\prev{\rel}^\myword(i_{\entf+1})}((\esphere[j_{\entf+1}],k))$. By guard $\test_2$ of T7, we obtain $\rho_{i_{\entf+1}}((\esphere[j_{\entf+1}],k)) = d_{i_{\entf+1}}^k$.


\item Assume $j_{\entf+1} \rel_{\entf+1}^{\esphere} j_{\entf}$. By T2, $\rel_{\entf+1} \in \dom(\source_i)$. Thus, there is (a unique) $i_{\entf+1}$ such that $i_{\entf+1} \rel_{\entf+1}^{\myword} i_{\entf}$.

By T3, we have $\esphere[j_{\entf+1}] \in q_{i_{\entf+1}}$. Moreover, $\lambda(j_{\entf+1}) = a_{i_{\entf+1}}$, and $\nu(j_{\entf+1}) = \mynu(i_{\entf+1})$.

Let $k \in [m]$ and $j \in U$. By condition T8, we have $\esphere[j_{\entf}] \in q_{i_{\entf}}$ implies $\update_{i_{e}}((\esphere[j],k)) = (\rel,(\esphere[j],k))$ for some $\rel \in \Sign$. Due to guard $\test_3$ of condition T7, we have $\rho_{\prev{\rel}^{\myword}(i_{\entf})}((E[j],k)) = \rho_{i_{\entf+1}}((E[j],k))$. We deduce $\rho_{i_\entf}((\esphere[j],k)) = \rho_{i_{\entf+1}}((\esphere[j],k))$.

Finally, let $k \in [m]$. We distinguish two cases. Suppose $\ptype{j_{\entf+1}} \neq \emptyset$. Then, $\update_{i_{\entf+1}}((\esphere[j_{\entf+1}],k)) = (\rel,(\esphere[j_{\entf+1}],k))$ where we let  $\rel = \rel_{\esphere[j_{\entf+1}]}$. Thus, $\rho_{i_{\entf+1}}((\esphere[j_{\entf+1}],k)) = \rho_{\prev{\rel}^\myword(i_{\entf+1})}((\esphere[j_{\entf+1}],k))$. By guard $\test_2$ of T7, we have  $\rho_{i_{\entf+1}}((\esphere[j_{\entf+1}],k)) = d_{i_{\entf+1}}^k$. If $\ptype{j_{\entf+1}} = \emptyset$, then $\rho_{i_{\entf+1}}((\esphere[j_{\entf+1}],k)) = d_{i_{\entf+1}}^k$ is due to the update $\update_{i_{\entf+1}}((\esphere[j_{\entf+1}],k)) = (k,0)$ (T8).

\end{itemize}
This concludes the proof of Lemma~\ref{cl:simulate}.
\qed
\end{proof}

By means of Lemma~\ref{cl:simulate}, we will show that spheres that are contained in states indeed occur in a data word. It will be used in combination with the following simple monotonicity fact, which follows easily from the definitions.


Next, we show that a sphere correctly simulates $\myword$ and vice versa, which concludes the correctness proof for $\Aut_B$.

For $i \in [n]$, let
$E_i = (\sphere_i,\sactive_i,\scolor_i)$ with $\sphere_i \df (U_i,(\rel^{\esphere_i})_{\rel \in \Sign},\labelone_i,\labeltwo_i,\scenter_i)$ be the unique extended sphere from $q_i$ such that $\scenter_i = \sactive_i$. In particular, $\sphere_i = \pi(q_i)$.

\begin{lemma}
For all $i \in [n]$, we have $\Sph{B}{\myword}{i} \cong \sphere_i$.
\end{lemma}

\newcommand{\hh}{\bar{h}}

\begin{proof}
For $\entf \in \{0,\ldots,\radius\}$, let $\entf$-$\sphere_i$ denote the $\entf$-sphere of $(U_i,(\rel^{\esphere_i})_{\rel \in \Sign},\labelone_i,\labeltwo_i)$ around $\scenter_i$, which is defined in the canonical manner.
We show, by induction, the following more general statement:
  \vspace{1.5ex}\\
  \hspace*{\fill}
   \begin{minipage}{0.9\linewidth}
     For every $\entf \in \{0,\ldots,\radius\}$, there is an isomorphism $h:
     \Sph{\entf}{\myword}{i} \rightarrow
     \entf$-$\sphere_i$ such that, for
     each $i' \in [n]$ with $\mydist^\myword(i,i') \le \entf$, we have $\esphere_i[h(i')] \in
     q_{i'}$.
   \end{minipage}
   \hspace{\fill}(*) \vspace{1.5ex}\\
We easily verify that (*) holds for $\entf = 0$. Now suppose there is an isomorphism $h: \Sph{\entf}{\myword}{i} \rightarrow \entf$-$\sphere_i$ with $\entf < B$.
We extend the domain of $h$ to elements $i'$ with $\mydist^\myword(i,i') = \entf+1$ as follows. Let $i_1,i_2 \in [n]$ such that $\mydist^\myword(i,i_1) = \entf$ and $\mydist^\myword(i,i_2) = \entf+1$. Let $\rel \in \Sign$. We distinguish several cases:
\begin{itemize}\itemsep=0.5ex
\item Suppose $i_1 \rel^\myword i_2$. Since $\mydist^\myword(i,i_1) < B$, we have $\mydist^\myword(\scenter_i,h(i_1)) < B$. By T6, there is $j_2 \in U_i$ such that $h(i_1) \mathrel{\rel^{E_i}} j_2$. Since $\esphere_i[h(i_1)] \in q_{i_1}$, we obtain, by T1, T4, and T7, $\lambda_i(j_2) = a_{i_2}$, $\nu_i(j_2) = \mynu(i_2)$, and $\esphere_i[j_2] \in q_{i_2}$.

\item Suppose $i_2 \rel^\myword i_1$. Similarly, due to $\mydist^\myword(i,i_1) < B$ and T5, there is $j_2 \in U_i$ such that $j_2 \rel^{E_i} h(i_1)$. Using T1, T3, and T7, we obtain $\lambda_i(j_2) = a_{i_2}$, $\nu_i(j_2) = \mynu(i_2)$, and $\esphere_i[j_2] \in q_{i_2}$.
\end{itemize}
We set $\hh(i_2) = j_2$ and $\hh(i') = h(i')$ for all positions $i'$ in $\Sph{\entf}{\myword}{i}$. In doing so, we extend the domain of $h$ to elements with distance $\entf+1$ from $i$. Note that this extension $\hh: \Sph{(\entf+1)}{\myword}{i} \rightarrow (\entf+1)$-$\sphere_i$ is well defined, i.e., $j_2$ is uniquely determined by $i_2$ and does not depend on the choice of $i_1$ or $\rel$: if, for $i_2$, we obtained distinct elements $j_2$ and $j_2'$, then $\esphere_i[j_2] \in q_{i_2}$ and $\esphere_i[j_2'] \in q_{i_2}$, which contradicts the definition of a state.

\smallskip

We show that we obtain a homomorphism $\hh: \Sph{(\entf+1)}{\myword}{i} \rightarrow
     (\entf+1)$-$\sphere_i$. Let $i_1,i_2 \in [n]$ such that $\mydist^\myword(i,i_1) = \mydist^\myword(i,i_2) = \entf+1$. Moreover, let $\rel \in \Sign$. Suppose  $i_1 \rel^\myword i_2$ (the case $i_2 \rel^\myword i_1$ is symmetric). We have $\esphere_i[\hh(i_1)] \in q_{i_1}$ and $\esphere_i[\hh(i_2)] \in q_{i_2}$. By T3 (or T4), this implies $\hh(i_1) \rel^{\esphere_i} \hh(i_2)$.

\smallskip

Next, we show that $\hh$ is surjective. Let $j_1,j_2 \in U_i$ and $\rel \in \Sign$ such that $\mydist^{\esphere_i}(\scenter_i,j_1) = \entf$, $\mydist^{\esphere_i}(\scenter_i,j_1) = \entf+1$, and $j_1 \rel^{\esphere_i} j_2$ (the case $j_2 \rel^{\esphere_i} j_1$ is similar). We have $\esphere_i[j_1] \in q_{h^{-1}(j_1)}$. By T4 and $q_{h^{-1}(j_1)} \not\in F_\rel$, there is $i_2 \in [n]$ such that $\mydist^\myword(i,i_2) = \entf+1$, $h^{-1}(j_1) \rel^\myword i_2$, and $\esphere_i[j_2] \in q_{i_2}$. We deduce that $\hh$ is surjective.

\begin{figure}[t]
\centering
\begin{minipage}{0.46\textwidth}
\centering
{
\scalebox{0.84}{
\begin{picture}(56,85)(-15,-3)
\unitlength=0.32em
  \gasset{Nframe=y,Nw=5.5,Nh=5.5,Nmr=8,ilength=4} 
  \node(e2)(0,0){}
  \node(x1)(0,20){}
  \node(x4)(0,40){}
  \node(x3)(0,60){}
  \node[Nw=4,Nh=4,Nmr=0](e)(31,10){}
  \node[Nw=4,Nh=4,Nmr=0](x2)(31,30){}
  \node[Nw=4,Nh=4,Nmr=0](x5)(31,50){}
  \put(1,0){
    \drawcurve[AHnb=0](2,20)(10,18)(15,15)(20,12)(28,10)
  }
  \put(1,0){
    \drawcurve[AHnb=0](2,0)(10,2)(15,5)(20,8)(28,10)
  }

  \put(1,20){
    \drawcurve[AHnb=0](2,20)(10,18)(15,15)(20,12)(28,10)
  }
  \put(1,20){
    \drawcurve[AHnb=0](2,0)(10,2)(15,5)(20,8)(28,10)
  }
  
  \put(1,40){
    \drawcurve[AHnb=0](2,20)(10,18)(15,15)(20,12)(28,10)
  }
  \put(1,40){
    \drawcurve[AHnb=0](2,0)(10,2)(15,5)(20,8)(28,10)
  }

  \gasset{Nframe=n,Nadjust=w,Nh=3,Nw=5,Nmr=0}
  \node(e2)(0,0){$i_2$}
  \node(e1)(0,20){$i_1$}
  \node(x3)(0,40){$i_1^{1}$}
  \node(x3)(0,60){$i_1^{2}$}
    \node(e)(31,10){$i$}
  \node(x3)(30,58){$\vdots$}
  \node(x3)(0,68){$\vdots$}

\node(e2)(-10,0){$\esphere_i[j_1]$}
\node(x1)(-10,20){$\esphere_i[j_1]$}
\node(x4)(-10,40){$\esphere_i[j_1]$}
\node(x6)(-10,60){$\esphere_i[j_1]$}

\node(e)(37,10){$\esphere_i$}
\node(x2)(37,30){$\esphere_i$}
\node(x5)(37,50){$\esphere_i$}
\end{picture}
}
}
\caption{$\hh$ is injective\label{fig:injective}}
\end{minipage}
~~~~~~
\begin{minipage}{0.46\textwidth}
\centering
{
\scalebox{0.84}{
\begin{picture}(56,85)(-15,-3)
\unitlength=0.32em
  \gasset{Nframe=y,Nw=5.5,Nh=5.5,Nmr=8,ilength=4} 
  \node(e2)(0,0){}
  \node(x1)(0,20){}
  \node(e1)(0,30){}
  \node(x4)(0,40){}
  \node(x3)(0,50){}
  \node(x6)(0,70){}
  \node[Nw=4,Nh=4,Nmr=0](e)(31,15){}
  \node[Nw=4,Nh=4,Nmr=0](x2)(31,35){}
  \node[Nw=4,Nh=4,Nmr=0](x5)(31,55){}
  \drawedge[ELside=r](e1,x1){$\rel^\myword$}
  \drawedge[ELside=r](x3,x4){$\rel^\myword$}
  \put(1,0){
    \drawcurve[AHnb=0](28,15)(25,16)(20,20)(15,25)(10,28)(2,30)
  }
  \put(1,20){
    \drawcurve[AHnb=0](28,15)(25,16)(20,20)(15,25)(10,28)(2,30)
  }
  \put(1,40){
    \drawcurve[AHnb=0](28,15)(25,16)(20,20)(15,25)(10,28)(2,30)
  }
  \put(1,0){
    \drawcurve[AHnb=0](28,15)(25,14)(20,10)(15,5)(10,2)(2,0)
  }
  \put(1,20){
    \drawcurve[AHnb=0](28,15)(25,14)(20,10)(15,5)(10,2)(2,0)
  }
  \put(1,40){
    \drawcurve[AHnb=0](28,15)(25,14)(20,10)(15,5)(10,2)(2,0)
  }
  \gasset{Nframe=n,Nadjust=w,Nh=3,Nw=5,Nmr=0}
  \node(e2)(0,0){$i_2$}
  \node(x1)(0,20){$i_2^{1}$}
  \node(e1)(0,30){$i_1$}
  \node(x4)(0,40){$i_2^{2}$}
  \node(x3)(0,50){$i_1^{1}$}
  \node(x6)(0,70){$i_1^{2}$}
  \node(e)(31,15){$i$}
  \node(x3)(30,63){$\vdots$}
  \node(x3)(0,78){$\vdots$}

\node(e2)(-10,0){$\esphere_i[j_2]$}
\node(x1)(-10,20){$\esphere_i[j_2]$}
\node(e1)(-10,30){$\esphere_i[j_1]$}
\node(x4)(-10,40){$\esphere_i[j_2]$}
\node(x3)(-10,50){$\esphere_i[j_1]$}
\node(x6)(-10,70){$\esphere_i[j_1]$}

\node(e)(37,15){$\esphere_i$}
\node(x2)(37,35){$\esphere_i$}
\node(x5)(37,55){$\esphere_i$}
\end{picture}
}}
\caption{$\hh^{-1}$ is a homomorphism\label{fig:infseq}}
\end{minipage}
\end{figure}

Let us show that $\hh$ is injective. Let $i_1,i_2 \in [n]$ such that $\mydist^\myword(i,i_1) = \mydist^\myword(i,i_2) = \entf+1$. Assume $i_1 \neq i_2$. We show that, then, $\hh(i_1) \neq \hh(i_2)$. Let $j_1 = \hh(i_1)$ and $j_2 = \hh(i_2)$. Assume, towards a contradiction, that $j_1 = j_2$. Furthermore, assume $i_1 < i_2$ (the other case is symmetric). In $E_i$, there are paths from $j_1$ to $\alpha$ and from $\alpha$ to $j_1$ that are simulated, in $\myword$, by paths from $i_2$ to $i$ and from $i$ to $i_1$, respectively. By Lemma~\ref{cl:simulate} and monotonicity of a signature, we can simulate these paths of $\esphere_i$ arbitrarily often in $\myword$. This yields an infinite descending chain $ \ldots < i_1^{2} < i_1^{1} < i_1 < i_2$ such that $\esphere[j_1] \in q_{i_1^{l}}$ and $d_{i_2}^k = d_{i_1}^k = d_{\smash{i_1^{l}}}^k$ for all $l \ge 1$ and $k \in [m]$. But this is a contradiction, as every word position has only finitely many smaller positions. The procedure is illustrated in Figure~\ref{fig:injective}.

Finally, we show that $\hh: \Sph{(\entf+1)}{\myword}{i} \rightarrow (\entf+1)$-$\sphere_i$ is actually an isomorphism. Let $j_1,j_2 \in U_i$ and $\rel \in \Sign$ such that $\mydist^{\esphere_i}(\scenter,j_1) = \mydist^{\esphere_i}(\scenter,j_2) = \entf+1$ and $j_1 \rel^{\esphere_i} j_2$. We show that this implies $\hh^{-1}(j_1) \rel^\myword \hh^{-1}(j_2)$. Set $i_1 = \hh^{-1}(j_1)$ and $i_2 = \hh^{-1}(i_2)$. Assume, towards a contradiction, that $i_1 \,\mathrel{{\not\!\!\rel}^\myword} i_2$. We have $j_1 \neq j_2$, $\esphere_i[j_1] \in q_{i_1}$, and $\esphere_i[j_2] \in q_{i_2}$. Due to the definition of the set of states of $\Aut_B$, this implies $i_1 \neq i_2$. Suppose $\nextwp{\rel}^\myword(i_1) < i_2$ (the other case is similar). Again, by Lemma~\ref{cl:simulate} and monotonicity of $\Int{\rel}{w}$, we can build an infinite descending chain $ \ldots < i_1^{2} < i_1^{1} < i_1 < i_2$ such that $\esphere[j_1] \in q_{i_1^{l}}$ for all $l \ge 1$ (cf.\ Figure~\ref{fig:infseq}). This is a contradiction.
\qed
\end{proof}

\section*{B.~~Comparison with Class Automata}
\label{sec:classautomata}

We compare \myautomata to class automata \cite{BL2010}, which have been shown to capture all (extended) XPath queries. Class automata are a smooth (undecidable) extension of data automata and, therefore, of \cmautomata. A class automaton is suitable to work over words (even trees) with multiple data values. It consists in a pair $(\Aut,\Baut)$ where $\Aut$ is a non-deterministic letter-to-letter transducer from the label alphabet $\Sigma$ to some working alphabet $\Gamma$, and $\Baut$ is a finite automaton over $\Gamma \times \{0,1\}^m$. A data word $(a_1,d_1) \ldots (a_n,d_n) \in \Sigma \times \{0,1\}^m$ is accepted if, for input $a_1 \ldots a_n$, there is some output $u_1 \ldots u_n \in \Gamma^\ast$ of $\Aut$ such that, for all $d \in \Data$, the word $(u_1,b_1) \ldots (u_n,b_n) \in (\Gamma \mathrel{\times} \{0,1\}^m)^\ast$ is accepted by $\Baut$. Hereby, $b_i^k = 1$ iff $d_i^k = d$.

We will show that, for $m=2$, class automata capture neither EMSO logic nor \ngautomata. Note that class automata do not depend on a signature. To allow for a fair comparison, we choose the simple signature $\CaSign = \{\edge{+1}\,,\edge{\sim}^1\,,\edge{\sim}^2\}$.

\begin{theorem}
  There is $L \in \ClassEMSO(\CaSign) \cap \ClassRCRA(\CaSign)$ such that $L$ cannot be recognized by any class automaton.
\end{theorem}

\begin{proof}
Let $\Sigma = \{a\}$ and $\Data = \N$. Using \cite{BL2010}, one can show that there is no class automaton that recognizes $L = [\{(a,1,1) \ldots (a,n,n)(a,1,1) \ldots (a,n,n) \mid n \ge 1\}]_{\CaSign}$. It is, however, easy to define an $\EMSO(\CaSign)$-sentence for $L$. We restrict to the construction of a \ngautomaton, which is very similar to the automaton from Example~\ref{ex:rCRA}. Here, we will need four registers, $r_1^k$ and $r_2^k$ for $k = 1,2$. The crucial difference is in the second phase, where we encounter a data value for the second time. We henceforth require that, at position $n+i$, the $k$-th data value $d_{n+i}^k$ is contained in register $r_1^k$ at $\cpred{\edge{\sim}^k}(n+i) = i$. The value $d_{n+i}^k$ is henceforth stored in  $r_1^k$ and has to coincide, at position $n+i+1$, with the contents of $r_2^k$ at position $\cpred{\edge{\sim}^k}(n+i+1) = i+1$.
\qed
\end{proof}
\fi

\end{document}